%% file: main.tex
\documentclass[acmsmall,nonacm,table]{acmart}
\include{packages}
\include{macros}

\begin{document}

\title{Dynamic Race Detection with O(1) Samples}

\author{Mosaad Al Thokair*}
\email{mosaada2@illinois.edu}
\affiliation{%
  \institution{University of Illinois at Urbana-Champaign}
  \city{Urbana}
  \country{USA}
}

\author{Minjian Zhang*\textsuperscript{†}}
\email{minjian2@illinois.edu}
\affiliation{%
  \institution{University of Illinois at Urbana-Champaign}
  \city{Urbana}
  \country{USA}
}

\author{Umang Mathur}
\email{umathur@comp.nus.edu.sg}
\affiliation{%
  \institution{National University of Singapore}
  \city{Singapore}
  \country{Singapore}
}

\author{Mahesh Viswanathan}
\email{vmahesh@illinois.edu}
\affiliation{%
  \institution{University of Illinois at Urbana-Champaign}
  \city{Urbana}
  \country{USA}
}
\footnotetext[1]{These authors contributed equally to this work.}
\footnotetext[2]{Corresponding author: minjian2@illinois.edu}

\input{abstract}



\maketitle

\input{intro}
\input{prelim}
\input{algo}

\input{experiments}

\input{related}
\input{conclusions}
\input{acks}
\appendix
\section{Appendix}
Trace characteristics, sampling statistics and races reported.
$\delta=0.1$ for \RND. \pacer uses sampling ratio=$3\%$.
\input{tables/table1}

\input{tables/table3}

\input{tables/table4}

\bibliographystyle{ACM-Reference-Format}
\bibliography{references}

\end{document}

%% file: packages.tex
\usepackage[utf8]{inputenc} 
\usepackage{amsmath, amsthm}
\usepackage{thmtools}
\usepackage[english]{babel}
\usepackage{pdfpages}
\usepackage{paralist}
\usepackage{graphicx}
\usepackage{subcaption}
\usepackage{caption}
\usepackage{longtable}
\usepackage{booktabs}
\usepackage{tabu}
\usepackage[referable]{threeparttablex}
\usepackage{varwidth}
\usepackage{multirow}
\usepackage{wrapfig}
\usepackage{enumerate}
\usepackage[inline]{enumitem}
\usepackage{microtype}
\usepackage{xifthen}
\usepackage{adjustbox}
\usepackage{pifont}

\usepackage[ruled,vlined,resetcount,linesnumbered]{algorithm2e}
\usepackage{parskip}

\usepackage{hyperref}
\usepackage{thm-restate}
\usepackage[capitalise]{cleveref}

\crefname{figure}{Figure}{Figure}
\crefformat{footnote}{#2\footnotemark[#1]#3}
\usepackage{comment}

\usepackage{todonotes}
\usepackage{tikz}
\usetikzlibrary{arrows,automata,shapes,decorations,decorations.markings,calc, matrix,decorations.pathmorphing, patterns,backgrounds,shapes.misc,arrows.meta}

\usepackage{bm}
\usepackage{pgfplots}
\usepackage{enumerate,paralist}
\usepackage{pbox}

\usepackage{appendix}
\usepackage{calc}
\usepackage{fp}
\usepackage{multirow}
\usepackage{pbox}

\usepackage[mathscr]{euscript}
\usepackage{mathtools}

\usepackage{footnote}
\usepackage{bbm}
\usepackage{multicol}

%% file: abstract.tex

\begin{abstract}
Happens before-based dynamic analysis is the go-to technique for detecting data races in large scale software projects due to the absence of false positive reports. However, such analyses are expensive since they employ expensive vector clock updates at each event, rendering them usable only for in-house testing. In this paper, we present a sampling-based, randomized race detector that processes only \emph{constantly many} events of the input trace even in the worst case. This is the first \emph{sub-linear} time (i.e., running in $o(n)$ time where $n$ is the length of the trace) dynamic race detection algorithm; previous sampling based approaches like {\pacer} run in linear time (i.e., $O(n)$). Our algorithm is a property tester for {\acrhb}-race detection --- it is sound in that it never reports any false positive, and on traces that are far, with respect to hamming distance, from any race-free trace, the algorithm detects an {\acrhb}-race with high probability. Our experimental evaluation of the algorithm and its comparison with state-of-the-art deterministic and sampling based race detectors shows that the algorithm does indeed have significantly low running time, and detects races quite often.
\end{abstract}

%% file: intro.tex

\section{Introduction}
\seclabel{intro}


A concurrent program is said to have a data race (or simply a race) if it has an execution in which a pair of threads access a shared memory location consecutively and in which one of the accesses writes a value to the shared memory location. Data races are one of the most common source of bugs in concurrent programs and are the cause of more serious problems like data corruption~\cite{boehmbenign2011,racemob,Narayanasamy2007}. Absence of data races is often a pre-requisite for the semantics of programs to be well defined~\cite{boehm-adve08,SoftwareErrors2009} and for compiler optimizations to be sound~\cite{sa08,sevcik11}. Sound dynamic data race prediction is a popular approach to identify such bugs in concurrent programs. Here one observes a single execution of the program, and the goal of the analysis is to see if the execution provides evidence for the presence of a data race in the program. This requires reasoning about alternate re-orderings of the events of the observed execution to determine if an execution with a data race is possible.

The simplest and most commonly used dynamic data race detection technique is based on Lamport's happens-before (\acrhb) partial order~\cite{lamport1978time}. The idea is to (implicitly) compute the {\acrhb} partial order on the events of the observed execution and check if there are a pair of conflicting accesses to a shared memory location that are unordered by {\acrhb}. Happens before based data race detection is known to be \emph{sound} --- the presence of {\acrhb} unordered memory accesses is proof that the program has a data race. Early vector clock based algorithms~\cite{Mattern1988,Fidge:1991:LTD:112827.112860} for happens before race detection have been improved over the years to the extent that it is the go-to data race detection in practice~\cite{threadsanitizer,threadsanitizerLLVM}. However, the analysis has high runtime costs~\cite{marino2009literace,bond2010pacer,racechaser} due to expensive metadata updates at each event, despite optimizations introduced~\cite{fasttrack,Pozniansky:2003:EOD:966049.781529}. This makes dynamic race detection suitable only for in-house testing, and makes an otherwise lucrative premise of data races \emph{exceptions}~\cite{elmas2007goldilocks,adve2010data} rather impractical.


One approach that attempts to reduce the analysis cost in {\acrhb} race detection uses \emph{sampling}~\cite{bond2010pacer}. The informal idea behind this approach is to sample some events of the observed trace, and analyze only the sampled subset of events to determine if the program has data races, with the hope being that the sampled subset will be small compared to the whole trace and that often it will be sufficient to find a race. Though the algorithm presented in~\cite{bond2010pacer} ({\pacer}) is experimentally shown to run faster than known deterministic race detection algorithms, its expected running time is linear and in the worst case it can be shown to be no better than a deterministic algorithm --- there are examples on which, with non-zero probability, {\pacer} will analyze \emph{all} the events in the trace.

The motivation behind this work is to explore the possibility of a sound race detection algorithm that analyzes, even in the worst case, a \emph{sub-linear} number of events in an execution (i.e., ``little $o$'' of the length of the trace), but nonetheless has provable mathematical guarantees of precision. The hope is that such an algorithm will help scale sound dynamic race detection beyond in-house testing.


To achieve our goal, we investigate the design of a \emph{property tester}~\cite{prop-test-book} for {\acrhb} race detection. A (one-sided) $(\epsilon,\delta)$-property tester for a decision problem $L$ is an algorithm $A$  that meets the following obligation: on an input $x \in L$, $A$ answers ``yes'' with probability $1$, while on an input that is ``$\epsilon$-far'' from any input in $L$, $A$ answers ``no'' with probability at least $1-\delta$. Notice that the definition of a property tester is based on a distance metric on input strings. The standard notion of distance used in property testing is hamming distance. Therefore, rephrasing and specializing to {\acrhb} race detection, we have the following. A property tester for {\acrhb} race detection is an algorithm $A$ such that on any {\acrhb} race-free execution $\tr$, $A$ answers ``yes'', and on any execution $\tr$ in which at least $\epsilon$ fraction of the events must be modified to obtain a race-free execution, $A$ answers ``no'' with probability at least $1-\delta$. Any sound and complete algorithm for {\acrhb} race detection is, by definition, a property tester for race detection since it distinguishes between race-free and racy executions. However, a property tester has weaker obligations and solves a decision problem approximately --- on executions $\tr$ that have an {\acrhb} race but at the same time are very close to a race-free execution, a property tester has no obligation to correctly classify them as having an {\acrhb} race. This flexibility has enabled computer scientists to design extremely fast, but nonetheless useful, algorithms for a variety of decision problems~\cite{prop-test-book} and has led to sub-linear algorithm design being a vibrant field of study for the past 25 years.


In this paper we present a $(\epsilon,\delta)$-property tester \RND~(\underline{\sf R}ace \underline{\sf P}roperty \underline{\sf T}ester) for {\acrhb}-race detection that provably examines only \emph{constantly many} events in the observed program execution. More precisely, let $t$ be an upper bound on the number of threads and $\lh$ be the maximum number of locks held at any point in the trace of a concurrent program. Our property tester analyzes only $\widetilde{O}(t+\lh)$ events of a trace $\tr$ and correctly classifies them as having an {\acrhb}-race with probability at least $1-\delta$, when $\tr$ is $\epsilon$-far from race-free executions. Here $\widetilde{O}(\cdot)$ hides constant, $\ln(1/\delta)$ and $(1/\epsilon)$ factors. Notice that the number of events examined by our property tester only depends on the parameters $t$ and $\lh$,
which are often very small, regarded as constants, 
and is independent of the length of the input trace $\tr$. 
Also note that, by design, \RND is a sound race detector ---- whenever it flags the presence of a race,
the execution has a real race.

Our property tester is a very simple, almost na\"{i}ve, algorithm, which maybe a feature when it comes to implementing it and deploying it in practice. It works as follows. If the input trace is ``short'' (defined precisely in \algoref{outline-property-tester}), run a deterministic race detector such as {\fasttrack}~\cite{fasttrack}. 
Otherwise, sample, uniformly at random, $O(\frac{\ln (1/\delta)}{\epsilon})$ sub-traces of input $\tr$, each of length $O(\frac{t+\lh}{\epsilon})$, and check that none of the sampled sub-traces contain an {\acrhb} race. If they do, the algorithm declares the input trace $\tr$ to be racy, and otherwise, declares it to be race-free. To check whether any of the sampled sub-traces contain an {\acrhb}-race, we could use any {\acrhb} race detection algorithm. In our experiments, we use the {\fasttrack} algorithm~\cite{fasttrack} that uses vector clocks and employs the \emph{epoch} optimization. 

The challenge, as for most randomized algorithms, is to prove that this simple algorithm is correct. This means we need to show that, if the input $\tr$, observed by running a multi-threaded program, is $\epsilon$-far from every race-free execution, then our algorithm will find an {\acrhb} race with high probability. The crux of our correctness proof is in the following observation. We show that any trace $\tr$ that is $\epsilon$-far from every race-free execution, has \emph{many, short} (of length $\widetilde{O}(t+\lh)$) sub-traces that contain an {\acrhb} race. Thus, by sampling a few different sub-traces independently, using standard arguments, one can show that the algorithm's answer is correct with high probability. 


We expect that the promise of a constant runtime overhead race detector will be useful for practitioners.
Given that the formal guarantee of our property tester is parameterized by $\epsilon$ and $\delta$, it is natural to ask how a practitioner should use our algorithm. After all, on the face of it, it seems like we need to know how far an observed execution $\tr$ is from race-free traces! \RND, like most dynamic techniques, is primarily an approach to find bugs. Therefore, our recommendation is to view $\epsilon$ and $\delta$ as adjustable parameters that a software developer can progressively decrease based on resource availability and past experiences rather than obligated parameters that one must decide for each program. If at any stage a race is discovered then debugging can begin. On the other hand, if no race is discovered even as $\epsilon$ and $\delta$ are decreased, then the software developer can be more confident about the reliability of the code based on the mathematical statements that back the correctness of the property tester. Finally, our experimental results show that even when the $\epsilon$ used in the algorithm is a poor measure of the actual distance of the input $\tr$ from race-free executions, the algorithm detects {\acrhb} races reasonably often.


\RND has been implemented. We have evaluated the performance of \RND on benchmark examples and compare it against the state-of-the-art deterministic (\fasttrack) and sampling-based (\pacer) \acrhb race detector, to see if the theoretical promises hold. We choose not to compare against techniques which employ a two-phase hybrid analysis~\cite{racemob,Choi02,razzer2019} because our innovations are primarily in dynamic analysis which is orthogonal to these approaches. Our techniques can be modularly plugged into hybrid race detection techniques to reduce the running time of the dynamic analysis phase and like RaceMob, can benefit from an additional static analysis phase (see \secref{related} for more details). 

Preliminary results suggest that \RND is a promising approach. When compared with \fasttrack and \pacer, \RND has the lowest running time among the $3$. Moreover, \RND's competitive advantage grows as the length of the trace increases. In fact, \RND's running time flattens out as the trace length grows in our experiments. Despite that, our results show that \RND reports a race quite often. This is especially true when considering traces that have a large proportion of race warnings --- when the number of race warnings reported by \fasttrack divided by trace length is at least $10^{-5}$, \RND detects races with at least the same probability if not better than \pacer. This is despite the fact that \RND in these experiments was run with a large value for $\epsilon$, namely $0.01$. Detailed experimental results are presented in \secref{experiments}.

%% file: prelim.tex

\section{Background and Preliminaries}
\seclabel{prelim}

In this section we discuss preliminary notations and also
recap relevant background on data race detection and property testing.

\subsection{Traces and Data Races}

\myparagraph{Concurrent Program Traces and Events}{
	The focus of our work is dynamic race detection,
	where one monitors the execution \emph{trace} of a concurrent program,
	observing \emph{events} generated by different threads,
	and analyzing it to infer the presence of a data race.
	Each event is labeled with a tuple $\ev{t, o}$ 
	(denoted simply as $e = \ev{t, o}$), 
	where $t$
	is the unique identifier of the thread that performs $e$
	and $o$ represents the operation associated with $e$.
	For our exposition, the operation $o$ can be one of~\footnote{We omit other synchronizations like forks and joins or wait-notify, for simplicity of presentation. It is straightforward to accommodate them, and all our results apply to the more general setting too. Further, our experiments do account for such events in the benchmarks.}:
	\begin{enumerate*}[label=(\alph*)]
		\item read/write access to a memory location $x$
		(i.e., $o = \rd(x)$ or $o = \wt(x)$), or
		\item lock-based synchronization --- acquire/release of a lock $\lk$
		(i.e., $o = \acq(\lk)$ or $o = \rel(\lk)$)
	\end{enumerate*}.
	We use the notation $\ThreadOf{e} = t$ and $\OpOf{e} = o$ for the event $e = \ev{t, o}$.
	A trace $\tr$ can thus be viewed as a sequence of such events (denoted $\events{\tr}$).
	We denote by $\threads{\tr}$, $\locks{\tr}$ and $\vars{\tr}$
	to denote the set of threads, locks and memory locations that appear in the trace $\tr$.
	We use $|\tr|$ to denote the length of $\tr$.
	}
	
	\myparagraph{Sub-traces}{
	Traces, as mentioned, are a sequence of events. 
	We will adopt the convention that the first event in the sequence has index $0$.
	Thus, a trace of length $n$ is of the form $\tr = e_0e_1\cdots e_{n-1}$. 
	The $i^\text{th}$ event of trace $\tr$ (namely $e_i$) will also be denoted as $\ith{\tr}{i}$. A \emph{sub-trace} $\substr{\tr}{i}{j} = e_ie_{i+1}\cdots e_{j-1}$ is the subsequence of $\tr$ of length $j-i$ from index $i$ to index $j-1$. When $j \leq i$, we adopt the convention that $\substr{\tr}{i}{j}$ is the empty sequence $\varepsilon$. The \emph{concatenation} of traces $\tr_1 = e_0\cdots e_{n-1}$ and $\tr_2 = f_0\cdots f_{m-1}$ is the sequence $e_0\cdots e_{n-1}f_0\cdots f_{m-1}$ of length $n+m$ and will be denoted by $\tr_1 \concat \tr_2$.
	}
	
	\myparagraph{Well formed Traces and Sub-traces}{
	Executions of concurrent programs, in addition to being a 
	sequence of events of the form described above, satisfy some properties. 
	A trace $\tr = e_0 \cdots e_{n-1}$ is \emph{well formed} if
	critical sections on the same lock do not overlap.
	That is, for every $j < n$, if $e_j = \ev{t,\rel(\lk)}$ releases lock $\lk$,
	then there is an $i < j$ such that $e_i = \ev{t,\acq(\lk)}$ and
	further, for every $i < k < j$, $\OpOf{e_k} \not\in \set{\acq(\lk), \rel(\lk)}$\footnote{We assume that locks are not re-entrant;
	all our results can nevertheless be extended in the presence of such locks.}.
	Henceforth, we will assume traces to be well formed.
	A sequence $\eta$ is a \emph{well formed sub-trace} if there is a well formed trace $\tr$ and indices $i$ and $j$ such that $\eta = \substr{\tr}{i}{j}$. Finally, in a well formed sub-trace $\eta = e_0e_1\cdots e_{n-1}$, we say that lock $\lk$ is \emph{held} by thread $t$ at $j$ (for $0 \leq j \leq n$) if either
	\begin{enumerate*}[label=(\alph*)]
	\item there is an $i < j$ such that $e_i = \ev{t,\acq(\lk)}$ 
	and for every $i < k < j$, $\OpOf{e_k} \neq \rel(\lk)$, or
	\item there is an $i \geq j$ such that $e_i = \ev{t,\rel(\lk)}$ 
	and for every $j \leq k < i$, $\OpOf{e_k} \neq \acq(\lk)$.
	\end{enumerate*}
	We say that lock $\lk$ is held at the beginning (resp. end)
	of a non-empty well-formed sub-trace $\eta$ if 
	$\lk$ is held at $0$ 
	(resp. $|\eta|$).	
	}


\myparagraph{Data Races}{
	A trace 
	is said to have a data race if two different threads access
	the same memory location without explicit synchronization in between.
	This is formalized in terms of 
	Lamport's Happens-Before (\acrhb) partial order~\cite{lamport1978time},
	which we recap next, while generalizing this notion 
	to well formed sub-traces.
	\begin{definition}[Happens-Before]
	For a well formed sub-trace $\tr$, the happens-before partial order induced by $\tr$,
	denoted $\stricthb{\tr}$, is the smallest binary relation on $\events{\tr}$ 
	such that for any two events
	$e_1 \neq e_2 \in \events{\tr}$, we have 
	$e_1 \stricthb{\tr} e_2$ if $e_1$ occurs before $e_2$ in $\tr$, and
	one of the following holds:
	\begin{description}
		\item[\it (program-order)] $\ThreadOf{e_1} = \ThreadOf{e_2}$.

		\item[\it (lock synchronization)] $e_1$ releases a lock $\lk$, which $e_2$ acquires (i.e., $\OpOf{e_1} = \rel(\lk)$ and $\OpOf{e_2} = \acq(\lk)$), or

		\item[\it (transitivity)] there is an event $e_3$ such that 
		$e_1 \stricthb{\tr} e_3 \stricthb{\tr} e_2$.
	\end{description}
	\end{definition}
	
	A pair of events $(e_1,e_2)$ of $\tr$ is said to be \emph{conflicting} if 
	$\ThreadOf{e_1} \neq \ThreadOf{e_2}$, both are memory access events on a common location (say) $x$ 
	with at least one of them being a write.
	An \emph{\acrhb-race} in $\tr$ is a pair $(e_1, e_2)$ of conflicting events in $\tr$ such that
	neither $e_1 \stricthb{\tr} e_2$ nor $e_2 \stricthb{\tr} e_1$.
	Finally, $\tr$ is said to have a data race if there is an \acrhb-race in it;
	otherwise $\tr$ is said to be race-free. \footnote{In general, the absence of HB-races does not imply the absence of predictive data races. The notion of race-freedom in this work means the absence of HB-races.}
	
	An important observation about the {\acrhb} partial order is that
	it is `context-free', i.e., 
	whether a pair of events is ordered by \acrhb only depends on 
	the events that appear between the two events in the trace.
	This is formalized in \propref{hb-context-free} 
	and will be crucially exploited by our randomized algorithm.
	\begin{proposition}
	\proplabel{hb-context-free}
	Let $\tr = e_0e_1\cdots e_{n-1}$ be a well formed sub-trace. 
	For any pair of indices $i < j$, $e_i \stricthb{\tr} e_j$ if and only if $e_i \stricthb{\substr{\tr}{i}{j+1}} e_j$.
	\end{proposition}
	
	\begin{proof}
	Follows from the definition of {\acrhb}.
	\end{proof}
}

\subsection{Dynamic Data Race Detection}

Data races can be detected in a streaming fashion,
by processing events one-by-one, updating metadata and 
checking for races at each event of interest, 
as shown in the general outline~\algoref{outline}.
Most of the algorithms known for detecting data races
dynamically~\cite{djit,Pozniansky:2003:EOD:966049.781529,elmas2007goldilocks,fasttrack}
adhere to this generic outline, and differ only on 
the precise details of the data maintained
by the algorithm or the implementation of the functions \init
and \checkAndUpdate.

\input{race-detection-outline}

\sloppy
Here, we discuss the most popular algorithm \djit~\cite{djit} which uses
vector clocks for assigning vector timestamps~\cite{Mattern1988,Fidge:1991:LTD:112827.112860}
to events and uses them to check for \acrhb-races;
\djit has further been optimized in subsequent works
including \djitp~\cite{Pozniansky:2003:EOD:966049.781529} and \fasttrack~\cite{fasttrack}.
A vector timestamp is a map $V : \threads{\tr} \to \nats$
that assigns a natural number to every thread of the trace $\tr$ being analyzed.
The ordering on two timestamps is defined as
$V_1 \cle V_2 \delequal \forall t \in \threads{\tr}, V_1(t) \leq V_2(t)$.
In essence, the \djit algorithm computes a vector timestamp $V_e$ 
for each event $e$ such that for any two events $e_1 \neq e_2$, 
$e_1 \hb{\tr} e_2$ iff $V_{e_1} \hb{\tr} V_{e_2}$.
However, instead of actually storing the timestamps of each event,
the algorithm uses \emph{vector clocks} to store a small number of timestamps. 
Overall, the algorithm maintains,
a vector clock $\Cc_t$, $\Ll_\lk$, $\Rr_x$ and $\Ww_x$ 
for every $t \in \threads{\tr}$, $\lk \in \locks{\tr}$
and $x \in \vars{\tr}$.

\input{vc-algo}

\algoref{vc-updates} summarizes how vector clocks are initialized
and updated --- the function \checkAndUpdate calls the appropriate
handler based on the operation performed in the event
(the timestamp $\bot$ is $\lambda u, 0$).
The lines~\ref{line:race-read}, \ref{line:race-write-1} 
and~\ref{line:race-write-2} perform the race detection checks.
We omit the increment `$\Cc_t \gets \Cc_t[\Cc_t(t)+1/t]$'
at the end of each handler.

While \acrhb-based dynamic analysis is considered
the go-to method for data race detection in practice~\cite{threadsanitizer,threadsanitizerLLVM},
it is known to add high runtime costs~\cite{marino2009literace,bond2010pacer,racechaser} due to 
expensive metadata updates at each event, despite optimizations
introduced~\cite{fasttrack,Pozniansky:2003:EOD:966049.781529}.
This makes dynamic race detection suitable only for in-house testing.

\subsection{Property Testing}
\seclabel{property-testing}

A property tester~\cite{prop-test-book} is an 
algorithm that solves a decision problem under a promise setting. 
Another way to think about it is that it solves a decision problem ``approximately''. 
It is typically a randomized algorithm. 
A property tester for a decision problem characterized by a language $L$ 
is an algorithm that provides the following guarantees: 
on a input $x \in L$ the algorithm answers \emph{yes} with high probability, 
and on an input $x$ that is ``far'' from anything in $L$, 
it answers \emph{no} with high probability. 
Thus, to define a property tester precisely, we need to identify 
a notion of distance between elements of the space of inputs. 
The most commonly used distance metric is \emph{hamming distance} which we define first.

\begin{definition}[Hamming Distance]
\deflabel{ham-dist}
For sequences $u,v \in \Sigma^*$ over alphabet $\Sigma$, 
the hamming distance between $u$ and $v$ ($\hamd(u,v)$) is defined as follows:
\[
	\hamd(u, v) = 
	\begin{cases}
		|\setpred{i}{\ith{u}{i} \neq \ith{v}{i}}| & \text{ if } |u| = |v| \\
		\infty & \text{otherwise}
	\end{cases}
\]
For $u \in \Sigma^*$ and $L \subseteq \Sigma^*$,
$
\hamd(u,L) = \inf_{v \in L} \hamd(u,v).
$
\end{definition} 

Note that for $u,v \in \Sigma^*$, either $\hamd(u,v) = \infty$, when $|u| \neq |v|$, or $\hamd(u,v) \leq |u| = |v|$, when $|u| = |v|$. 

Having defined the notion of distance between an input and a language, we can define precisely what a property tester is. In this paper, we only consider \emph{one-sided} testers and so we specialize the definition to this case.
\begin{definition}[Property Tester]
\deflabel{prop-test}
A (one-sided) $(\epsilon,\delta)$ property tester for a problem $L$ is a randomized algorithm $A$ such that on any input $x$, $A$'s output satisfies the following property.
\begin{enumerate}[label=(\alph*)]
\item If $x \in L$, $A(x) = \mbox{yes}$ with probability $1$.
\item If $\hamd(x,L) \geq \epsilon |x|$, $A(x) = \mbox{no}$ with probability at least $1 - \delta$.
\end{enumerate}
\end{definition}

We note some observations about \defref{prop-test}. 
On inputs $x$ that are not in $L$ but are ``close'' (i.e. $\hamd(x,L) < \epsilon |x|$), the property tester may answer either yes or no, without violating its obligation. In this sense, a property tester is an approximate algorithm for a decision problem. Second, since we are considering one-sided testers, we can arrive at the following conclusions about an input based on the tester's response. If $A(x) = \mbox{no}$ then we can conclude that $x \not\in L$. On the other hand, if $A(x) = \mbox{yes}$ then we cannot conclude anything definite about the membership of $x$ in $L$.



%% file: race-detection-outline.tex

\small
\begin{algorithm}[t]
\Input{Trace $\tr$}
\BlankLine
\init{} \; 
\lFor{$e$ in $\tr$}{ \linelabel{iterate_over_events}
    \checkAndUpdate{$e$}  
}
\caption{\textit{Outline for dynamic data race detection}}
\algolabel{outline}
\end{algorithm}
\normalsize

%% file: vc-algo.tex

\small
\begin{algorithm}[t]
\vspace*{-\multicolsep}
\begin{multicols}{2}
\myfun{\init{}}{
	\For{$t \in \threads{}$ $\cdot$}{
		$\Cc_t$ := $\bot[1/t]$
	}
	\For{$x \in \vars{}$ $\cdot$}{
		$\Rr_x$ := $\bot$; 
		$\Ww_x$ := $\bot$
	}
	\For{$\lk \in \locks{}$ $\cdot$}{
		$\Ll_\lk$ := $\bot$
	}
}

\myhandler{\acqhandler{$t$, $\lk$}}{
	$\Cc_t\gets \Cc_t \mx \Ll_\lk$ \;
}

\myhandler{\relhandler{$t$, $\lk$}}{
	$\Ll_t\gets \Cc_t$\;
}

\myhandler{\rdhandler{$t$, $x$}}{
	\check $\Ww_x \cle \Cc_t$ \; \linelabel{race-read}
	$\Rr_x\gets \Cc_t$
}

\myhandler{\wthandler{$t$, $x$}}{
	\check $\Rr_x \cle \Cc_t$ \; \linelabel{race-write-1}
	\check $\Ww_x \cle \Cc_t$ \; \linelabel{race-write-2}
	$\Ww_x\gets \Cc_t$
}
\vspace*{-\multicolsep}
\end{multicols}
\normalsize
\caption{\textit{Vector clock updates}}
\algolabel{vc-updates}
\SetAlgoInsideSkip{medskip}
\end{algorithm}
\normalsize

%% file: algo.tex

\section{A Property Tester for Race Detection}
\seclabel{algo}

Our randomized algorithm \RND for dynamic race detection is simple and straightforward:
\begin{enumerate}
\item Sample uniformly at random $\numsam$ sub-traces, each having length $\samlen$, of the observed trace $\tr$.
\item If any of the sampled sub-traces has a data race, declare $\tr$ to have a data race.
\item If none of the sampled sub-traces have a data race, declare $\tr$ to be race-free.
\end{enumerate}
To complete the description of \RND, we need to answer the following questions. How many sub-traces should the algorithm sample (parameter $\numsam$)? What should the length of sampled sub-traces be (parameter $\samlen$)? Obviously these parameters are set to ensure that the resulting algorithm is a property tester for race detection. The correctness proof is the most critical piece, and also the most technically challenging part of the algorithm itself. Our description in this section will use vague terms like ``many'', ``very far'', ``short'', ``long'' etc. These terms will be precisely characterized by parameters in our formal theorem and lemma statements, but our use of these informal terms in the text helps illuminate the main ideas behind the correctness argument without getting lost in the technical details.

Before we begin outlining the correctness proof, let us examine our algorithm template and establish some straightforward facts. First notice that no matter what values we set for parameters $\numsam$ and $\samlen$, the algorithm is \emph{sound} --- if the algorithm declares a trace $\tr$ to have a data race then it does indeed have a data race. This is because of the ``context-free'' property of {\acrhb}-races articulated in \propref{hb-context-free} --- whether events $e_1$ and $e_2$ of $\tr$ are in {\acrhb}-race only depends on the events in the sub-trace of $\tr$ that starts with $e_1$ and ends with $e_2$. This means that a race-free execution $\tr$ will be declared to be correct with probability $1$. Thus, to establish correctness, our obligation is to find values for $\numsam$ and $\samlen$ that ensure that on traces which are very far from any race-free execution, the algorithm discovers a sub-trace with an {\acrhb}-race with high probability. Second, to check if a sampled sub-trace has a data race, we could use any algorithm to check for {\acrhb}-races~\cite{djit,Pozniansky:2003:EOD:966049.781529,elmas2007goldilocks,fasttrack,ziptrack2018,kulkarnifinegrained2021}. In our implementation, we use {\fasttrack}~\cite{fasttrack}, but this could be replaced by any improvements to {\acrhb}-race checking to yield faster running times.

\subsection{Proof of Correctness}

Let us now present an overview of our correctness proof. It crucially relies on our main theorem (\thmref{main}) which says that if the input trace $\tr$ is far from any race-free execution (with respect to hamming distance), then there are many short sub-traces that contain an {\acrhb}-race. The measure that characterizes ``short'' in this theorem will be taken to be the value of $\samlen$. Observe that if \thmref{main} guarantees that presence of many $\samlen$-length sub-traces that have a race, then by analyzing a randomly chosen $\samlen$-length sub-trace guarantees that we will discover a race with some probability. Therefore if we repeat this experiment a few times (namely $\numsam$ times), we can ensure that we discover a race with high probability. Here the number of samples $\numsam$ will be set based on the chance that a single $\samlen$-length sub-trace is racy, using standard counting arguments. The main theorem (\thmref{main}) itself is established by first observing that a trace $\tr$ which is far from any race-free execution has many (not necessarily short) \emph{disjoint} sub-traces that have an {\acrhb}-race. This is the content of \lemref{disjoint-race} which is proved as follows. We show that if $\tr$ has very few disjoint sub-traces that are racy, then $\tr$ can be transformed by changing very few events into a race-free execution. Since we know $\tr$ is far, we can conclude that it has many disjoint races. Finally, to show that few disjoint racy sub-traces means closeness to a race-free execution, we need \lemref{hb-concatenation} which proves that any pair of race-free sub-traces $\tr_1$ and $\tr_2$ can be combined into a larger race-free sub-trace, provided we paste a short sub-trace $\mu$ between $\tr_1$ and $\tr_2$.

Having provided an overview of our proof, we are ready to present the technical details. We start with a technical lemma that shows that for any well formed sub-traces $\tr_1$ and $\tr_2$, there is a \emph{short} trace $\mu$ such that the concatenated sub-trace $\tr_1\mu\tr_2$ is well formed with the property that every event in $\tr_1$ is {\acrhb}-before any event in $\tr_2$. This will be used later to show that in traces that are far from race-free executions, there are many disjoint racy sub-races. 
\begin{lemma}
\lemlabel{hb-concatenation}
Let $\tr_1$ and $\tr_2$ be well formed subtraces over threads $T$. Let $\lh$ be an upper bound on the number of locks held at the end of $\tr_1$ and at the beginning of $\tr_2$. There exists a sub-trace $\mu(\tr_1,\tr_2)$ such that $|\mu(\tr_1,\tr_2)| \leq 4|T| + 2\lh$, and  $\tr = \tr_1\mu(\tr_1,\tr_2)\tr_2$ is well formed. Moreover, for any events $e_1 \in \events{\tr_1}$ and $e_2 \in \events{\tr_2}$, we have $e_1 \stricthb{\tr} e_2$. 
\end{lemma}

\begin{proof}
Let $\s{LH}_1$ be the set of locks held at the end of $\tr_1$ and let $\s{LH}_2$ be the set of lock held at the beginning of $\tr_2$. Notice that $|\s{LH}_1| \leq \lh$ and $|\s{LH}_2| \leq \lh$. Without loss of generality, let us assume that $\lk_* \in \locks{\tr_1} \cup \locks{\tr_2}$. $\mu(\tr_1,\tr_2)$ is the following sequence of events in the given order.
\begin{enumerate}
\item \itmlabel{lstr} If $\lk_*$ is held by thread $t_*$ at the end of $\tr_1$ then start with $\ev{t_*,\rel(\lk_*)}$.
\item \itmlabel{order} \sloppy For each thread $t \in T$, add the sequence $\ev{t,\acq(\lk_*)}\ev{t,\rel(\lk_*)}$. After adding such a sequence for each thread, \emph{repeat} this sequence again. That is, once again, for every  thread $t \in T$, add the sequence $\ev{t,\acq(\lk_*)}\ev{t,\rel(\lk_*)}$.
\item \itmlabel{lh1} For each lock $\lk \in \s{LH}_1 \setminus \set{\lk_*}$ that is held by thread $t$ at the end of $\tr_1$, add the event $\ev{t,\rel(\lk)}$.
\item \itmlabel{lh2} For each lock $\lk \in \s{LH}_2$ held by thread $t$ at the beginning of $\tr_2$, add the event $\ev{t,\acq(\lk)}$.
\end{enumerate}
Observe that the number of events added in step \itmref{lstr} + step \itmref{lh1} is at most $\lh$. Similarly, the number events added in step \itmref{lh2} is at most $\lh$. Finally, the number of events added in step \itmref{order} is $4|T|$. Putting all of this together, proves that $|\mu(\tr_1,\tr_2)| \leq 4|T| + 2\lh$. 
Next, the order in which events are added in $\mu(\tr_1,\tr_2)$ ensures that in $\tr$ each lock is held by at most one thread at any given time, which means that $\tr$ is well formed. Moreover, the set of locks held at the beginning of $\tr_2$ in $\tr$ is exactly $\s{LH}_2$. Finally, the events added in step \itmref{order} ensure that for any events $e_1 \in \events{\tr_1}$ and $e_2 \in \events{\tr_2}$, we have $e_1 \stricthb{\tr} e_2$.
\end{proof}

\begin{remark}
It is worth noting an important consequence of \lemref{hb-concatenation} that we will exploit in our proof. Observe that the sub-trace $\mu(\tr_1,\tr_2)$ constructed in the proof has no data access events. Thus, if $\tr_1$ and $\tr_2$ are race-free, then so is $\tr_1\mu(\tr_1,\tr_2)\tr_2$.
\end{remark}

For the rest of this section, let fix a set of threads $T$, a set of locks $L$, and a set of memory locations $M$. Let $\rcf$ be the set of all well formed traces over $T$, $L$ and $M$ that are race free. That is,
\[
\rcf = \setpred{\tr \mbox{ race free}}{\threads{\tr} \subseteq T,\ \locks{\tr} \subseteq L,\ \vars{\tr} \subseteq M}.
\]
Observe that for any trace $\tr$, $\hamd(\tr,\rcf) \leq |\tr|$. This is because we can always pick $\eta \in \rcf$ of the same length as $\tr$, by ensuring that either $\eta$ only events performed by a single thread, or has no write events, etc.

Let us also assume that $\lh$ is an upper bound on the number of locks held at any point in a trace; in the worst case $\lh = |L|$, but typically $\lh$ is much smaller than $|L|$. Finally, let us fix $m = 4|T| + 2\lh$. 


\lemref{hb-concatenation} allows one to show that if a trace $\tr$ is very far from the set $\rcf$ (as measured by parameter $\epsilon$), then there are many \emph{disjoint} sub-traces of $\tr$ that contain a pair of events that are in {\acrhb}-race. In other words, if $\tr$ is far from any race-free trace, then there are many disjoint witnesses that demonstrate that $\tr$ has a race.

\begin{lemma}
\lemlabel{disjoint-race}
Let $\tr$ be a trace of length $n$ such that $\hamd(\tr,\rcf) \geq \epsilon n$. 
There is an integer $u \geq \frac{\epsilon n}{m}$ and 
an increasing sequence of indices 
$0 = i^1_1 < i^2_1 < i^1_2 < i^2_2 < \cdots < i^1_u < i^2_u \leq n$ 
of length $2 \cdot u$ such that
each sub-trace $\substr{\tr}{i^1_j}{i^2_j}$ ($1 \leq j \leq u$) has an {\acrhb}-race.
\end{lemma}

\begin{proof}
Let us construct an increasing sequence of indices as follows. Take $i^1_1 = 0$. The remaining indices are inductively defined as follows. Assuming $i^1_1, i^2_1, \ldots i^1_j$ have been defined. Then,
\[
i^2_j = \min \setpred{k \leq n}{\substr{\tr}{i^1_j}{k} \mbox{ has a race}}.
\]
In the above equation, if the set over which we are taking a minimum is empty (i.e., $\substr{\tr}{i^1_j}{n}$ is race-free) then our construction of the sequence ends. Next, assuming $i^1_1, i^2_1, \ldots i^1_j, i^2_j$ are defined, we take $i^1_{j+1} = i^2_j+m-1$, provided $i^2_j + m -1 < n$; again if $i^2_j + m - 1 \geq n$, then we stop the construction.

Notice that by definition, our sequence is increasing and each sub-trace $\substr{\tr}{i^1_j}{i^2_j}$ has an {\acrhb}-race. To complete the proof of the lemma, all we need to argue is that the sequence we have constructed is long, i.e., if $i^2_u$ is the last index constructed by the above sequence, then $u \geq \frac{\epsilon n}{m}$. We will use the fact that $\hamd(\tr,\rcf) \geq \epsilon n$ to establish this.

Suppose we have constructed the sequence $0 = i^1_1 < i^2_1 < \cdots < i^1_u < i^2_u$ as above. Since we stopped at $i^2_u$, it means that either $i^2_u + m - 1 \geq n$ or $\substr{\tr}{i^2_u+m-1}{n}$ is race-free. Let us define the sub-trace $\tr_j$ as $\substr{\tr}{i^1_j}{i^2_j-1}$. Notice by definition of $i^2_j$ this means that $\tr_j$ is race free. Consider the trace $\tr'$ defined as follows.
\[
\tr' = \tr_1 \concat \mu(\tr_1,\tr_2) \concat \tr_2 \concat \mu(\tr_2,\tr_3) \cdots \tr_u \concat \mu_*.
\]
Here $\mu(\tr_j,\tr_{j+1})$ is the sequence guaranteed by \lemref{hb-concatenation} for $\tr_j$ and $\tr_{j+1}$, and $\mu_*$ is defined as follows: if $i^2_u+m-1 \geq n$, then $\mu_*$ is some race-free trace of length $n - i^2_u + 1$ and if $\substr{\tr}{i^2_u+m-1}{n}$ is race-free then 
$\mu_* = \mu_*' \concat \substr{\tr}{i^2_u+m-1}{n}$ where $\mu_*' = \mu(\tr_u,\substr{\tr}{i^2_u+m-1}{n})$.
Without loss of generality, we will assume that sub-traces of the form $\mu(\tr_j,\tr_{j+1})$ guaranteed by \lemref{hb-concatenation} are of length \emph{exactly} $m$ --- if they are shorter, we can pad them with events.

Notice that $|\tr'| = n = |\tr|$ and $\tr'$ is (by construction) race-free because of the remark after \lemref{hb-concatenation}. Moreover $\hamd(\tr,\tr') \leq um$, since at most $m$ events are changed in each $\mu$-sub-trace. Since $\tr'$ is race-free and $\hamd(\tr,\rcf) \geq \epsilon n$, we have $um \geq \epsilon n$ which means that $u \geq \frac{\epsilon n}{m}$. This completes the proof of the lemma.
\end{proof}

\lemref{disjoint-race} guarantees the presence of many disjoint, racy sub-traces. However, that by itself is not enough to get an efficient property tester for data race detection. In particular, \lemref{disjoint-race} provides no bounds on the length of the racy sub-traces it identifies. If we do not strengthen \lemref{disjoint-race}, the only bound we can get on the sample length $\samlen$ in our template algorithm would be $|\tr|$, which would make our property tester no more efficient than a deterministic race detector. Our main theorem, established next, shows that, for sufficiently long traces, there are many racy sub-traces of \emph{short} length, when a trace $\tr$ is far from any race-free trace. The proof uses \lemref{disjoint-race}. This will enable us to bound $\samlen$ and get good asymptotic bounds.

\begin{theorem}
\thmlabel{main}
Let $\tr$ be a trace of length $n$ such that $\hamd(\tr,\rcf) \geq \epsilon n$. In addition, let $n \geq (12m)/\epsilon$. Then there are at least $2(\epsilon n)/15$ sub-traces of $\tr$ of length $4m/\epsilon$ that contain an {\acrhb}-race.
\end{theorem}

\begin{proof}
Let $0 = i^1_1 < i^2_1 < i^1_2 < \cdots i^1_u < i^2_u \leq n$ be the increasing sequence of indices guaranteed by \lemref{disjoint-race} such that each subtrace $\substr{\tr}{i^1_j}{i^2_j}$ has a data race. Consider the set 
\[
\s{long} = \setpred{\substr{\tr}{i^1_j}{i^2_j}}{i^2_j - i^1_j \geq (2m/\epsilon)}.
\]
Since each element of $\s{long}$ is a sub-trace of $\tr$, the sum of the lengths of such sub-traces is $\leq n$. On the other hand, each sub-trace in $\s{long}$ is of length at least $2m/\epsilon$ and so the sum of the lengths is at least $\frac{2m |\s{long}|}{\epsilon}$. Putting it together, we get
\[
n \geq \sum_{\rho \in \s{long}} |\rho| \geq \frac{2m |\s{long}|}{\epsilon} \quad \Rightarrow \quad |\s{long}| \leq \frac {\epsilon n}{2m}.
\]
Since $u \geq (\epsilon n)/m$, we have $|\s{short}| \geq (\epsilon n)/(2m)$, where
\[
\s{short} = \setpred{\substr{\tr}{i^1_j}{i^2_j}}{i^2_j - i^1_j < (2m/\epsilon)}.
\]

The above counting argument guarantees that the number of disjoint short racy traces is large. However, we can improve the bound further if we allow sub-traces to overlap. This improvement will help improve the running time of our property tester in turn. 

Consider a sub-trace $\eta = \substr{\tr}{i^1_j}{i^2_j} \in \s{short}$. Let $|\eta| = s$; we know $s \leq (2m/\epsilon)$. Each such sub-trace $\eta$ (unless $j = 1$ or $j = u$) is a sub-trace of $(4m/\epsilon - s)$ sub-traces of $\tr$ of length $4m/\epsilon$. The reason is because any sub-trace of $\tr$ of length $4m/\epsilon$ that starts at a position in the interval $[i^2_j - 4m/\epsilon, i^1_j]$ contains $\eta$. Since $s \leq 2m/\epsilon$, we have each such $\eta$ is contained in at least $4m/\epsilon - s \geq (2m/\epsilon)$ sub-traces of length $4m/\epsilon$. Note, that each sub-trace of length $4m/\epsilon$ that contains such an $\eta \in \s{short}$ has a data race. On the other hand, any sub-trace $\rho$ of $\tr$ of length $4m/\epsilon$ can contain at most $(5m/\epsilon)/m$ sub-traces from $\s{short}$. This can be argued as follows. Suppose $\rho$ contains $a$ sub-traces in $\s{short}$. Since each sub-trace in $\s{short}$ is separated by $m$ positions (see proof of \lemref{disjoint-race}), the sum of the lengths of all $\s{short}$ sub-traces plus their intervening gaps is at least $(a-1)m$. Now $|\rho| = 4m/\epsilon$. Thus, $(a-1)m \leq 4m/\epsilon$, which means that $a \leq (4/\epsilon)+1 \leq 5/\epsilon$. In other words, each sub-trace of $\tr$ of length $4m/\epsilon$ contains at most $5/\epsilon$ of the sub-traces in $\s{short}$. Putting these observations together, we see that the number of sub-traces $\rho$ of $\tr$ of length $4m/\epsilon$ that contain a data race is at least
\[
\frac{2m}{\epsilon}\cdot\frac{\epsilon}{5}\cdot [|\s{short}| - 2] \geq \frac{2m}{5}\left[ \frac{\epsilon n}{2m} - 2\right].
\]
In the above equation, ``$-2$'' is to discount $\substr{\tr}{i^1_1}{i^2_1}$ and $\substr{\tr}{i^1_u}{i^2_u}$ if they belong to $\s{short}$. Assuming $n \geq (12 m)/\epsilon$, we have $(\epsilon n)/(2m) - 2 \geq (\epsilon n)/(3m)$. Thus, the number of sub-traces of length $4m/\epsilon$ that contain a data race is at least
$\frac{2m}{5}\cdot\frac{\epsilon n}{3m} = \frac{2\epsilon n}{15}.$
\end{proof}


\thmref{main} helps complete the description of our algorithm \RND. Our property tester will pick sub-traces of length $4m/\epsilon$, the parameter used in \thmref{main}. There are $n$ such sub-traces, since each starting position identifies such a sub-trace. From \thmref{main} it follows that the probability that a random sub-trace of length $\samlen = 4m/\epsilon$ has a data race (when $\tr$ is $\epsilon$-far from $\rcf$), is at least $2\epsilon/15$. If we pick $\numsam = \frac{15\ln (1/\delta)}{2\epsilon}$, the probability we will not detect is data race is at most
\[
(1 - 2\epsilon/15)^{\numsam} < e^{-\ln (1/\delta)} = \delta.
\]

\subsection{Pseudocode for \RND}

Let us conclude this section by presenting a pseudo-code for our property tester (\algoref{outline-property-tester}). Recall that the algorithm samples $\numsam = \frac{15\ln (1/\delta)}{2\epsilon}$ sub-traces of input $\tr$ of length $\samlen = 4m/\epsilon$, and checks if any of the sampled sub-traces have an {\acrhb}-race. Notice that sampling a sub-trace of length $\samlen$ is the same as picking a starting index $i$ with the understanding that the sampled sub-trace is $\substr{\tr}{i}{i+k}$. Thus, sampling $\numsam$ sub-traces is the same as picking $\numsam$ starting indices (\lineref{sample}). Consider two sub-traces $\substr{\tr}{i_1}{i_1+k}$ and $\substr{\tr}{i_2}{i_2+k}$ of $\tr$ that overlap. That is, wlog $i_1 < i_2 < i_1+k$. Notice that by the definition of {\acrhb} partial order, if the sub-trace $\substr{\tr}{i_1}{i_2+k}$ is race-free then both $\substr{\tr}{i_1}{i_1+k}$ and $\substr{\tr}{i_2}{i_2+k}$ are race-free. Thus, we can merge sampled sub-traces that are overlapping without sacrificing on our ability to detect races. Therefore, in \lineref{merge}, we merge the overlapping sub-traces to get a smaller set of sampled sub-traces, but with the possibility of the sampled sub-traces being longer that $\samlen$. This step reduces the total number of events our algorithm will process. After this initial pre-processing step, the algorithm proceeds as follows. When an event is the start of a sampled sub-trace, the meta-data is reset so that there is a fresh start to race detection. In addition, whenever an event is in our sampled sub-trace we call the function \checkAndUpdate which in turn calls the appropriate handler in \algoref{vc-updates} based on the operation performed by the event. When an event is not in any of our sampled sub-traces, no checking and meta-data updates take place.

\input{property-testing-algo-outline}

We remark that \RND can be easily extended to synchronisation primitives such as those due to atomic/volatile accesses, as well as synchronizations such as fork-join, wait-notify or barriers. This is because for such synchronizations, the “context-free” property (\propref{hb-context-free}) holds and the classic \fasttrack algorithm can be easily extended for such synchronizations (as also observed in \cite{fasttrack}).

%% file: property-testing-algo-outline.tex

\small
\begin{algorithm}[t]
\Params{$\epsilon, \delta \in [0, 1]$}
\BlankLine
\Input{Trace $\tr$}
\BlankLine
\init{} \; 
$m \gets 4|T|+2h$ ;
$\samlen \gets 4m/\epsilon$; 
$\numsam \gets 15 \ln(1/\delta)/(2\epsilon)$ \;
\If{$|\tr| < 12m/\epsilon$}{
      Run \algoref{outline} on $\tr$
}
\Else{
       $I \gets$ Sample $\numsam$ indices in $[0,n-\samlen]$ \linelabel{sample} \;
       $S \gets$ \merge($\set{\substr{\tr}{i}{i+k}}_{i \in I}$) \linelabel{merge}\;
       \For{$e$ in $\tr$}{ \linelabel{iterate}
         	\If{$e$ is the start of a sub-trace in $S$}{
	               \reset{}
	        }
                \If{$e$ is in an sub-trace in $S$}{
   	               \checkAndUpdate{$e$} 
   	        }
        }
}
\caption{\textit{Property Tester for Checking \acrhb-races}}
\algolabel{outline-property-tester}
\end{algorithm}
\normalsize

%% file: experiments.tex

\section{Experimental Evaluation}
\seclabel{experiments}

\input{experiment_intro.tex}
\input{experiment_implementation}
\input{benchmarks_and_setup}
\input{Technical}
\input{observation}

\input{epsilon}

%% file: experiment_intro.tex

We evaluated the practical feasibility of our property tester
by implementing \RND and comparing against the go-to deterministic
race detection algorithm \fasttrack due to Flanagan and Freund~\cite{fasttrack}, 
and against \pacer~\cite{bond2010pacer}, which is the state-of-the art sampling based
race detection algorithm. 
Each of these tools has a different philosophy and solves a slightly different problem --- \fasttrack processes every event in the trace and reports \emph{all} {\acrhb} races; \pacer promises that every {\acrhb} race has an equal, non-zero probability of being reported; and in contrast, \RND is engineered to sample just enough to ensure that we can mathematically prove that \emph{at least one race} will be reported with high probability, when the trace is far from being race free. Their comparison is not a like-for-like comparison. As a consequence, the experiments in this section are not to suggest that one tool is better than another. They are there to merely help one understand the likely performance of \RND on practical programs: is the running time low as promised by the theory, how often does it report at least one race, does it work only when there are many races, what types of races arise in practice and are there some that \RND will always fail on. If we only report the performance of \RND it becomes difficult to gauge how reasonable it is, and therefore, we report the performance of both \fasttrack and \pacer to serve as a baseline.

The rest of this section is organized as follows. After explaining our implementation and setup (\secref{implementation}), and characteristics of our extracted traces (\secref{trace-char}), we present our experimental analysis in two parts. In the first part (\secref{one-race}), we present experiments that help understand how effective \RND is in reporting at least one race. We look at how the running time of \RND changes with trace characteristics like length and number of threads + locks held. Next, our theoretical analysis only guarantees reporting at least one race with high probability whenever the observed execution is far from a race-free execution. Therefore, we ask how does \RND perform on traces with very few races, since the number of races in a trace is an upper bound on the distance of a trace from a race-free trace. Finally, recall that \RND relies on sampling short sub-traces and can only detect races between events that are not far apart. So we ask, how often do traces have races that are close by, and how does \RND perform on traces where most or all races are between events that far apart (when compared to the length of sub-traces sampled by \RND). The second part of our analysis (\secref{multiple-races}) reports on the performance of \RND as measured by the number of races detected. Note that \RND is not engineered to report all races or most races or even every race with some probability. It only guarantees reporting \emph{some} race with high probability, when the trace is far from being race-free. Thus, these experiments are not consistent with the design of \RND, but our objective here is to understand if there are some types of races that will escape detection with \RND.

\rmv{
The experiment section is organized as two parts:  

\begin{description}
    \item[(Technical Guarantees)] In \secref{algo}, we showed that even though our property tester samples only constantly many events, it provides the following mathematical guarantee of correctness --- with high probability $1-\delta$,
	\RND will report a race when the analyzed execution trace
	is sufficiently far from any race-free execution trace (as measured by $\epsilon$) with constant analysis time. 
	Our evaluation here examines
	how \RND fares on the following two metrics, directly answering
	how our approach is expected to perform in practice.
	In doing so, we will also examine how these experimental results relate to our theoretical guarantee. 
\begin{description}
    \item[(\cost)] Developers are more likely to use 
    a race detector, especially in a deployed setting, when the analysis time is low.
  	Our analysis in~\secref{algo} guarantees that when the constant $m = 4|T| + 2h$ is small and the parameters $\epsilon$
  	and $\delta$ are not incredibly small, \RND will encounter only constant time for expensive operations.
    The primary goal of our experimental evaluation is to measure the actual
    analysis cost of our property tester, and understand how it varies as parameters are changed.
    We also want to compare \RND against the baseline deterministic algorithm \fasttrack,
    and against the sampling-based technique of \pacer. 
    {\pacer} provides a guarantee of detecting each race with a fixed probability, but has no bounds on the number of events it processes.

    \item[(Precision)] While \RND is sound --- meaning a warning is reported only when it corresponds to a real \acrhb-race --- it is prone to missing races. The reason is that it samples only a small subset of events
    in an execution trace, missing races occurring outside the sampled sub-traces.
    Our theoretical analysis, however, guarantees reporting at least one race with high probability whenever
    the observed execution is far from a race-free execution. It is important to note that this guarantee
    is conditioned upon a good estimation of the \emph{\raciness} of the trace, as defined in the following sections.
    In our evaluation, we gauge our algorithm's race detection capability 
    when varying the likelihood of any race and compare it with \pacer's and \fasttrack's.
    \item[(Traces with Long Races)] \RND's high precision is conditioned and it can only detect "short races" such that both events can be sampled in the single sampling phase. We therefore examines the number of long races that \RND is unlikely to detect for each benchmark and examines how \RND behaves in the benchmarks where the long races dominate. Surprisingly, \RND is still able to detect at least one race for most of these benchmarks and even for the one where no short race existed(because sampling period can overlap and form a larger period).
\end{description}

 \item[(Detecting Various Races)]  
 The technique section does not provide any guarantee on the characteristics of races that \RND can detect, however, for most of the races, including the rare ones, we expect \RND to have a good chance to detect them with some probabilities. We examines this by evaluating the following three metrics:
    \begin{enumerate}
        \item ability to detect number of warnings 
        \item ability to expose racy memory locations 
        \item ability to expose racy source code locations 
    \end{enumerate}
\ucomment
{We also compared \RND against \pacer on these metrics. However, as \RND and \pacer adopt different philosophies, and the tools are expected to behave differently, the section does not conclude that one tool outperforms the other on any of the metrics but is designed to provide some observations to address developer's concern.   }
 \end{description}
 }

%% file: experiment_implementation.tex

\subsection{Implementation and Setup}
\seclabel{implementation}

We have implemented our algorithm \RND, the \fasttrack algorithm
and \pacer's sampling algorithm in Java.
Our implementation is designed to run all three algorithms on the exact same trace to allow for a fair comparison and reduce noise in the results introduced by the runtime thread-scheduler.

\myparagraph{\pacer}{
\pacer is the state-of-the art sampling based race detection 
technique for detecting data races.
At a high level, \pacer partitions the observed execution
trace into \emph{sampling} and \emph{non-sampling} phases, similar to \RND.
In each sampling period, \pacer monitors all events, by performing metadata 
updates as in~\fasttrack, similar to \RND.
But unlike \RND, the expected total size of the sampling periods is $r\cdot n$ 
where $r$ is a sampling rate set by \pacer, and $n$ is the total size
of the execution trace.
The more stark difference is that \pacer also performs
metadata updates on (a subset of) events in a non-sampling period ---- all synchronization
operations, as well as all memory locations that were accessed in prior sampling periods
are tracked.
This means that, in general, \pacer effectively can analyze a large number of events,
much more than what is determined by its proportionality constant $r$.
On the other hand, our approach guarantees that the number of events
analyzed by \RND over the course of the entire execution is bounded by
a constant which is determined by the number of threads, lock
nesting depth and the chosen values for the parameters $\epsilon$ and $\delta$.


The publicly available implementation of \pacer~\cite{pacer-tool} is built on top of 
Jikes RVM-3.1.0, which is only compatible with an old version of Java 1.6.0.
Further, this implementation does not support comparing the same execution trace
against different runtime techniques.
A distinguishing feature of the Jikes-RVM implementation of \pacer's algorithm
is the use of the runtime garbage collector to implement a periodic random sampler.
In our implementation, we simulate the effect of a periodic sampler by
invoking our sampler in a periodic fashion.
Our implementation of \pacer in \rapid closely mimics
the algorithm's description in the original paper~\cite{bond2010pacer},
including the use of version clocks, shallow copies and optimizations in
vector clock joins. 
}

%% file: benchmarks_and_setup.tex

\input{tables/table2}

\myparagraph{Benchmarks}{
Our benchmark programs are primarily derived
from prior works which evaluate the performance of different race detection techniques~\cite{MathurTreeClocks2022,SyncP2021,fasttrack,rv2014,wcp2017}.
These include Java benchmarks from the DaCaPo benchmark suite~\cite{DaCapo2006}, 
Java Grande Forum~\cite{JGF2001}, microbenchmarks 
from~\cite{vonPraun03} and SIR~\cite{SIR2005}, and OpenMP benchmarks
derived from DataAccelerator~\cite{schmitz2019dataraceonaccelerator},
DataRaceBench~\cite{liao2017dataracebench},
OpenMP source code repository~\cite{dorta2005openmp}
and the NAS parallel benchmarks~\cite{nasbenchmark}
as well as from HPC applications including CORAL benchmarks~\cite{coral,coral2}, 
ECP proxy applications~\cite{ecp},
and Mantevo project~\cite{mantevo}.
}

\myparagraph{Experimental Setup}{
The execution traces of Java benchmarks were logged using
the {\textsc{RoadRunner}}~\cite{flanagan2010roadrunner} dynamic analysis framework
and traces from OpenMP benchmarks were logged using ThreadSanitizer~\cite{threadsanitizer}.
There were $90$ programs in the benchmark suite and we ran some programs on two different inputs. This resulted in a total of $149$ benchmark traces, after filtering out short traces (with $<1M$ events).
We analyzed each of these traces against the three algorithms we have implemented.
We use the parameter combination $(\epsilon=0.01, \delta=0.1)$
for \RND and set \pacer's sampling
rate to be $3\%$ as suggested in~\cite{bond2010pacer}.
On each trace and for each combination of parameters, 
we ran {\fasttrack} $20$ times, while {\RND} and {\pacer} $50$ times, 
to account for randomization and variations in system load; {\fasttrack} was run fewer times because it is a deterministic algorithm. 
We report the average performance of each tool in our tables and graphs.
Our experiments were conducted on machine with 2.6GHz 64-bit Linux machine,
using Java-1.8 as the JVM and 30GB heap space.}

\subsection{Trace Characteristics}
\seclabel{trace-char}

Since the number of traces in our collection is large, we will
summarize the data by clustering traces according to the values of several parameters.
We choose two metrics for clustering traces --- lengths of the traces (or number of events),
and the value of $m = 4|T| + 2h$ which governs the length of
each sample \RND extracts, and thus the total number of sampled events.
The choice of the second metric allows us to focus on intrinsically \emph{similar} traces at the same time.
Table~\ref{table:table2} summarizes our set of traces, in overall terms, as well as with the details of each $m$-based cluster.
Observe that the total number of events go as high as $2.8$ billion,
the average trace length is $297$ million and the median trace length is around $135$ million. 
The trace lengths are diverse overall, as well as within each cluster.
In total there are $5$ benchmarks that are race free. Categorized by trace length, there are $30$ benchmarks in $(0, 100M]$, $70$ in $(100M, 200M]$, $24$ in $(200M, 400M]$, $14$ in $(400M, 700M]$ and $11$ in $(700M, 3B]$~\footnote{$100M$ denotes $100$ million or $10^8$, while $1B$ denotes $1$ billion or $10^9$.}. We also provide a table with more detailed information for the benchmarks in the appendix.



%% file: tables/table2.tex

\begin{table*}[t]
\caption{
Characteristics of traces.
We aggregate benchmark traces in $6$ clusters, based on the values of parameter $m = 4|T| + 2h$.
Column 1 shows the range of $m$ values in each cluster and Column 2 shows the number of traces in each such cluster.
Columns 3-9 show the average, min, max and different percentiles of the lengths of the traces in each cluster.
The last row shows these metrics for the entire dataset.
}
\label{table:table2}

\centering
\scalebox{0.8}{
\begin{adjustbox}{center}
\renewcommand{\arraystretch}{1.3}
\begin{tabular*}{1.16\columnwidth}{!{\VRule[1pt]}c!{\VRule[1pt]}c!{\VRule[1pt]}c|c|c|c|c|c|c|c!{\VRule[2pt]}}
\specialrule{1pt}{0pt}{0pt}
1 & 2 & 3 & 4 & 5 & 6 & 7 & 8 & 9\\

\specialrule{1pt}{0pt}{0pt}
\rowcolor[HTML]{EFEFEF} 
\cellcolor[HTML]{EFEFEF} $m$ 
& \cellcolor[HTML]{EFEFEF} Num. of
& \multicolumn{7}{c!{\VRule[1pt]}}{\cellcolor[HTML]{EFEFEF}{Trace Length Distribution}} \\

\cmidrule[1pt]{3-9}
\rowcolor[HTML]{EFEFEF} 
\cellcolor[HTML]{EFEFEF} (range) 
& \cellcolor[HTML]{EFEFEF} traces
& \cellcolor[HTML]{EFEFEF} \; Average \;
& \cellcolor[HTML]{EFEFEF} \; Min \;
& \cellcolor[HTML]{EFEFEF} \; 20 \%-ile \;
& \cellcolor[HTML]{EFEFEF} \; 40 \%-ile \;
& \cellcolor[HTML]{EFEFEF} \; 60 \%-ile \;
& \cellcolor[HTML]{EFEFEF} \; 80 \%-ile \;
& \cellcolor[HTML]{EFEFEF} \; Max \;
 \\
\specialrule{1pt}{0pt}{0pt}
\textsf{(0, 29]}        & 12  & 650.2M & 40.0M& 134.3M& 291.7M& 539.6M& 607.8M& 2.8B\\
\textsf{(29, 59]}       & 10  & 391.3M & 1.0M& 47.0M& 124.6M& 253.4M& 323.0M& 2.4B \\
\textsf{(59, 69]}       & 57  & 165.9M & 3.1M& 104.9M& 112.3M& 135.0M& 168.9M& 1.3B \\
\textsf{(69, 95]}       & 14  & 541.1M & 11.7M& 90.2M& 265.4M& 533.9M& 771.1M& 1.6B
 \\
\textsf{(95, 231]}      & 49  & 294.2M & 11.7M& 106.8M& 132.9M& 175.1M& 360.0M& 2.1B \\
\textsf{(231, 1000]}    & 7   & 158.7M & 39.1M& 65.2M& 177.8M& 199.9M& 207.5M& 259.1M\\
\specialrule{1pt}{0pt}{0pt}
All traces              & 149 & 297M & 1.0M & 102M & 127M & 172M & 349M & 2.8B \\
\specialrule{1pt}{0pt}{0pt}
\end{tabular*}
\end{adjustbox}
}
\label{tab:time2}
\end{table*}

%% file: Technical.tex
\subsection{Detecting at Least One Race}
\seclabel{one-race}

\input{scalability}

\input{constant}
\input{Precision}

\input{longrace}

%% file: scalability.tex

\RND is designed to approximately solve the decision problem of race detection, i.e., answer the question whether an observed trace has a race. The innovation in the algorithm is to identify how little sampling will still allow one to mathematically guarantee reporting a race (with high probability) in a trace that is far from being race-free. In this section, we explore how effective the theoretical claims are in practice by answering the following questions.
\begin{itemize}
    \item \RND is designed to sample minimally so that its running time is low. Does this hold in practice? We answer this in the affirmative.
    \item Theoretical claims about \RND's correctness guarantee detecting a race only when the trace is $\epsilon$-far from being race-free. In practice, how does \RND perform when the distance of the observed trace from race-free traces is much less than $\epsilon$? Does \RND successfully report races in such traces? We find that \RND does report races often even when the observed trace is very close to being race free.
    \item \RND samples short sub-traces and checks for races within these sub-traces. Thus, a race pair that is far apart will not be detected by \RND since the sub-traces sampled will not have both the events forming the race pair. How often do traces contain race pairs that are not far apart? Second, is \RND able to report races, when almost all the races are at distance greater than the sample length used by \RND? In our benchmark suite, we observe that most traces have races between nearby events. Moreover, even in traces where almost all races are far apart, \RND successfully reports races often.
\end{itemize}

\myparagraph{Running Time}{
Our implementation of each of \fasttrack, \pacer and \RND analyzes trace logs,
and we measure the running time of race detection by simply measuring the time taken by each of the algorithms to analyze each trace.
The running times are computed as the average of the time taken during each run of an 
algorithm on a given trace. 
In order to be able to present our results for the large benchmarks visually, we cluster traces by their lengths. For each cluster, we compute the weighted average of the running times and speedups(\fasttrack as baseline) for \fasttrack, \pacer, and \RND, where the weight for a trace is the reciprocal of its length.

\begin{figure}[h]
\begin{subfigure}{0.45\textwidth}
\includegraphics[width=\textwidth]{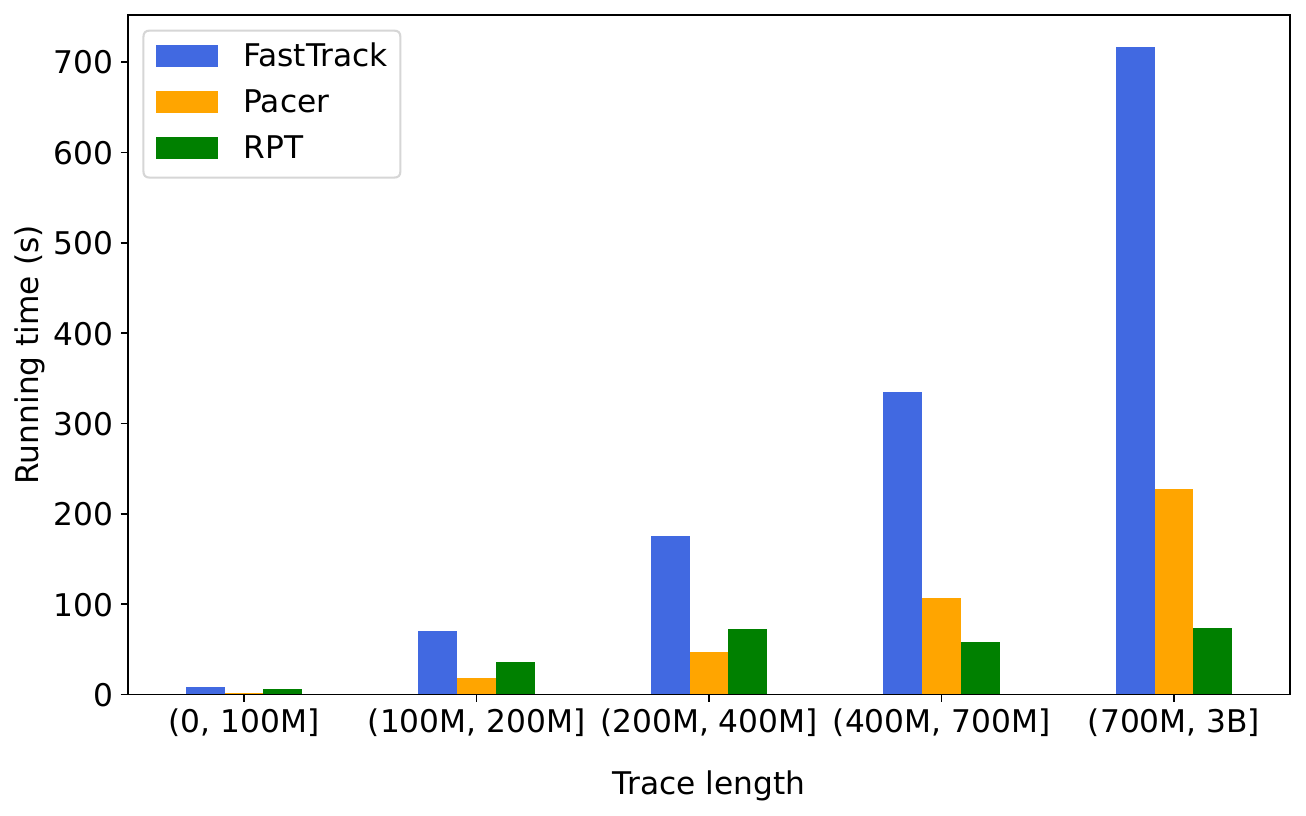}
\vspace{-0.1in}
\end{subfigure}
\begin{subfigure}{0.45\textwidth}
\includegraphics[width=\textwidth]{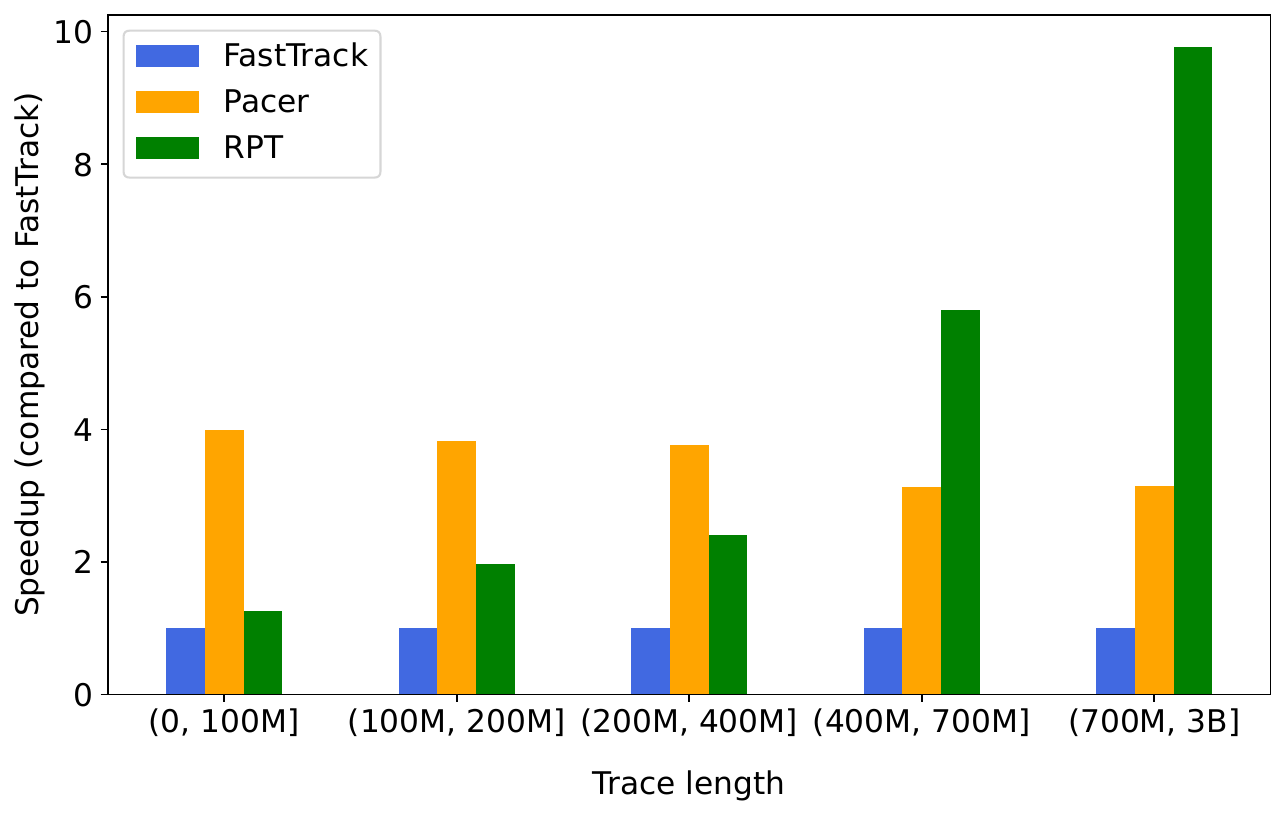}
\vspace{-0.1in}

\end{subfigure}
\caption{Weighted average running time (left) as function of trace length (range) and weighted average speedup over \fasttrack (right) as function of trace length(range) }

\figlabel{runningtime-bucketed-by-tracelength}

\end{figure}

Not surprisingly, the running time of \RND is pretty low,
and in fact, the lowest for the larger traces. 
Further, \pacer is also significantly faster than {\fasttrack}. 
As trace lengths increase, \RND's competitive advantage over \pacer and \fasttrack 
becomes more significant, and \RND's running time does not grow as fast. 

\begin{figure*}[t]
\begin{subfigure}{0.48\textwidth}
    \includegraphics[width=\textwidth]{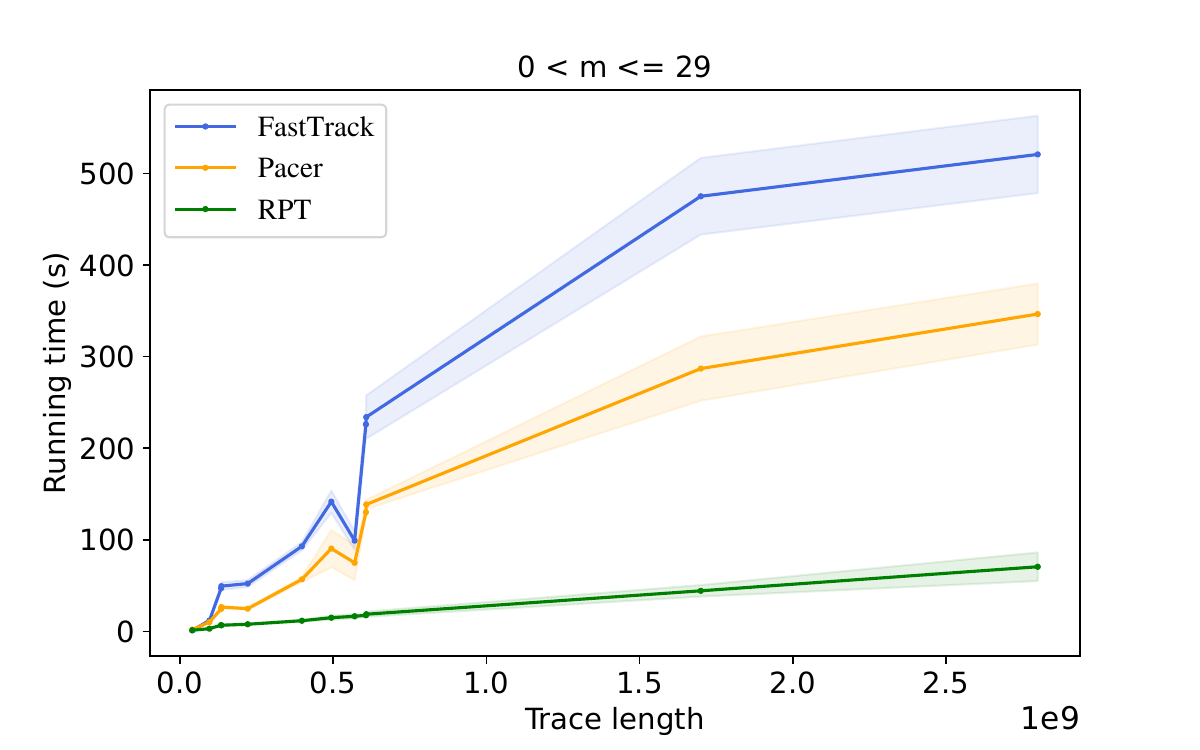}
    \caption{Running times for traces with $m \in (0, 29]$}
\end{subfigure}
\begin{subfigure}{0.48\textwidth}
    \includegraphics[width=\textwidth]{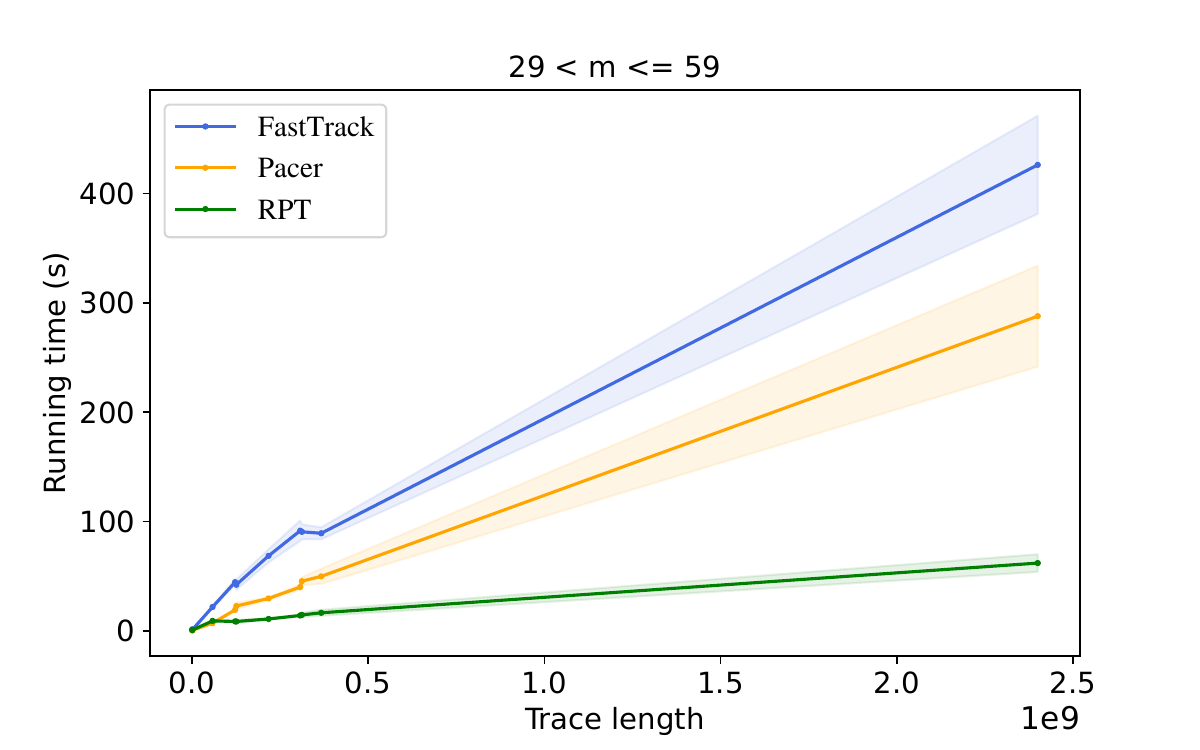}
    \vspace{-0.2in}
    \caption{Running times for traces with $m \in (29, 59]$}
\end{subfigure}
\begin{subfigure}{0.48\textwidth}
    \includegraphics[width=\textwidth]{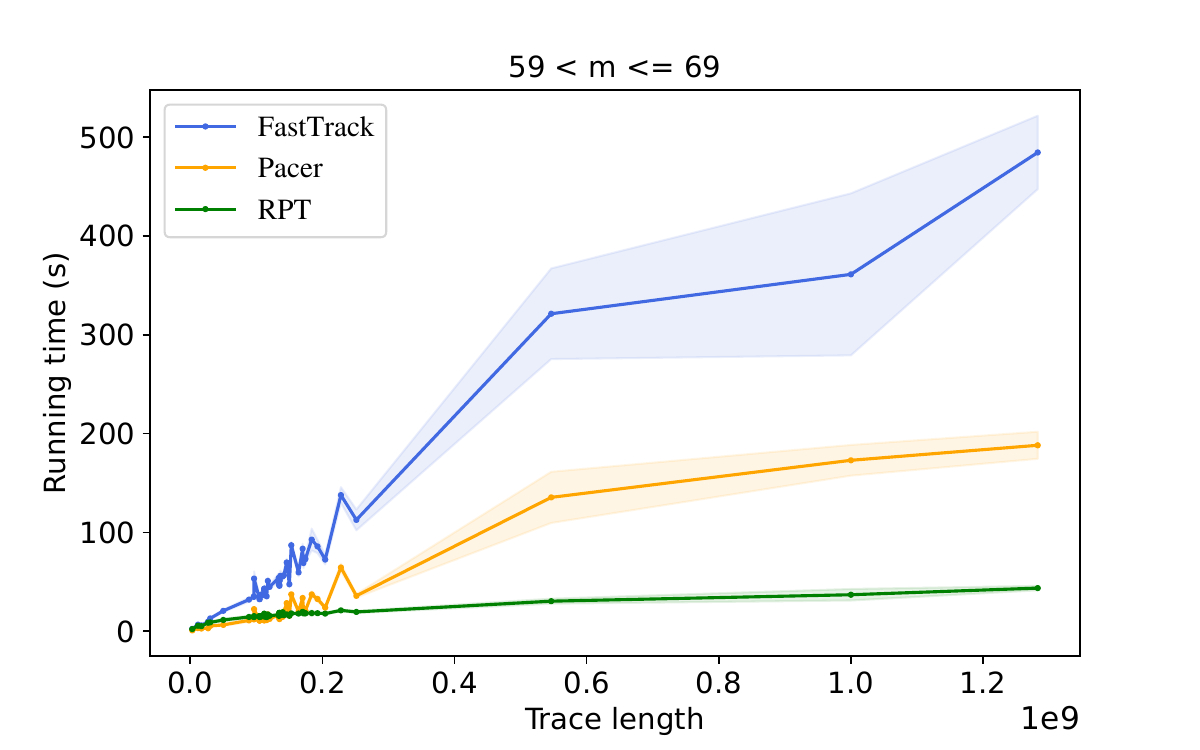}
    \caption{Running times for traces with $m \in (59, 69]$}
\end{subfigure}
\begin{subfigure}{0.48\textwidth}
    \includegraphics[width=\textwidth]{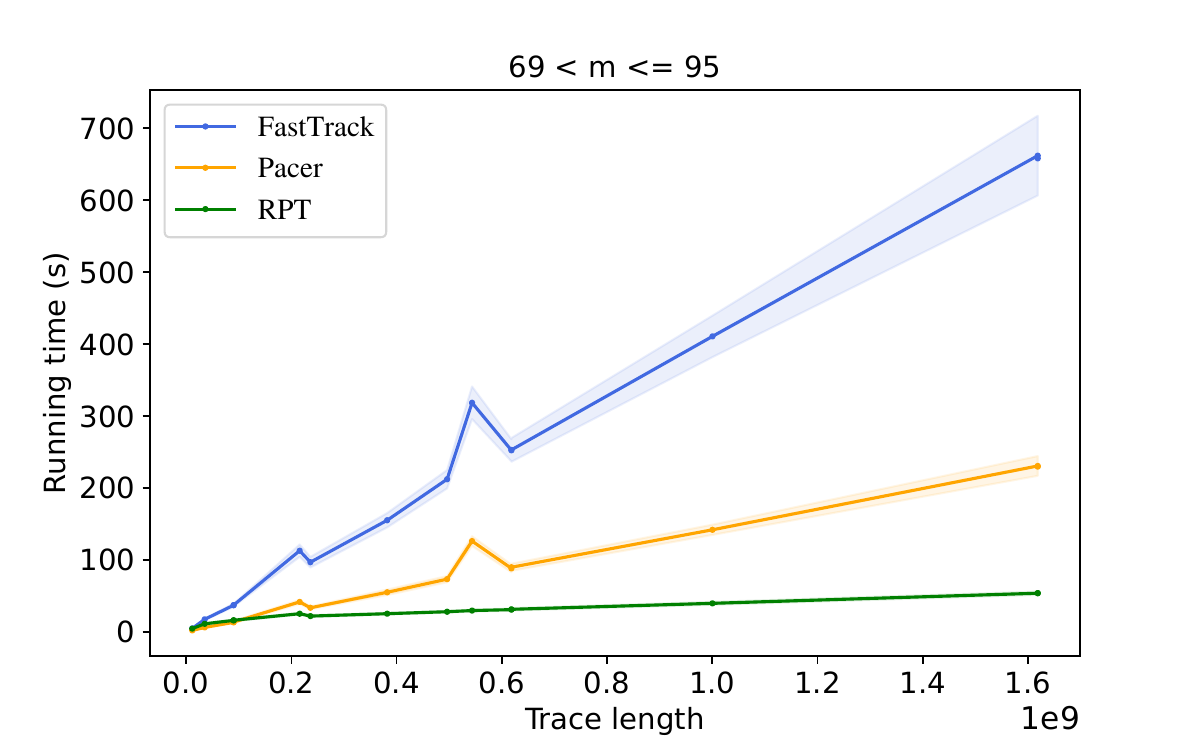}
    \vspace{-0.2in}
    \caption{Running times for traces with $m \in (69, 95]$}
\end{subfigure}
\begin{subfigure}{0.48\textwidth}
    \includegraphics[width=\textwidth]{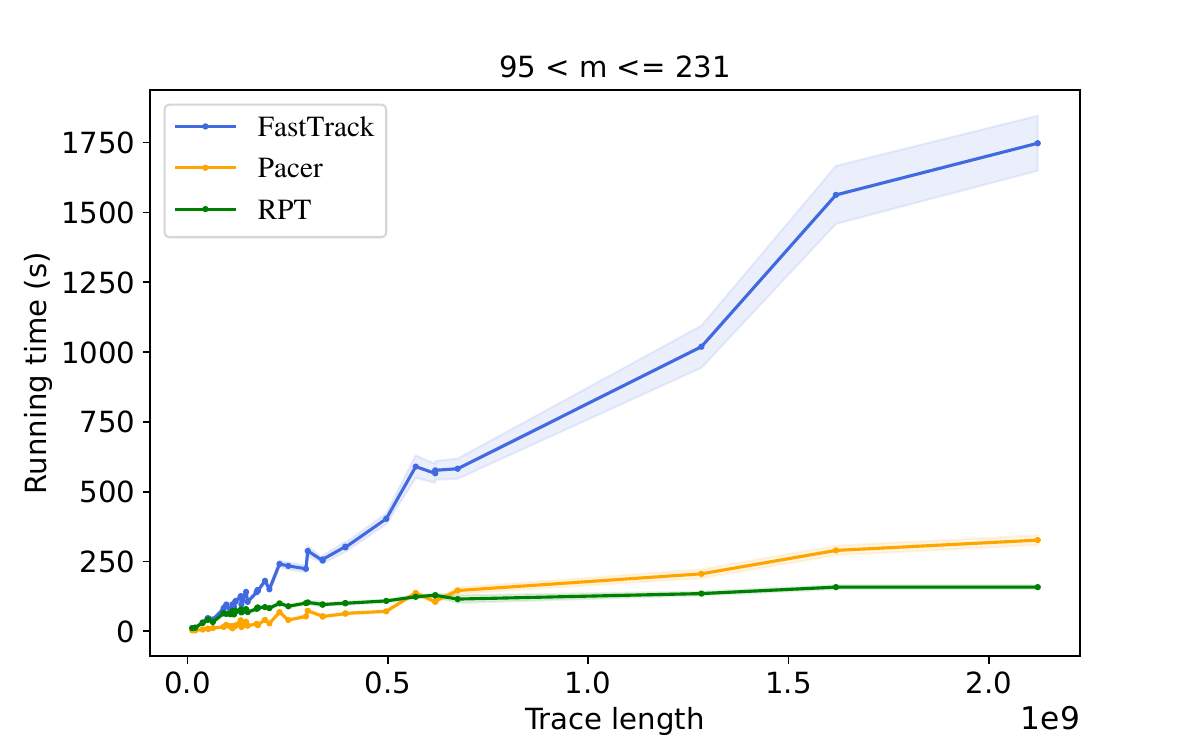}
    \caption{Running times for traces with $m \in (95, 231]$}
    \figlabel{trace-sizes}
\end{subfigure}
\begin{subfigure}{0.48\textwidth}
    \includegraphics[width=\textwidth]{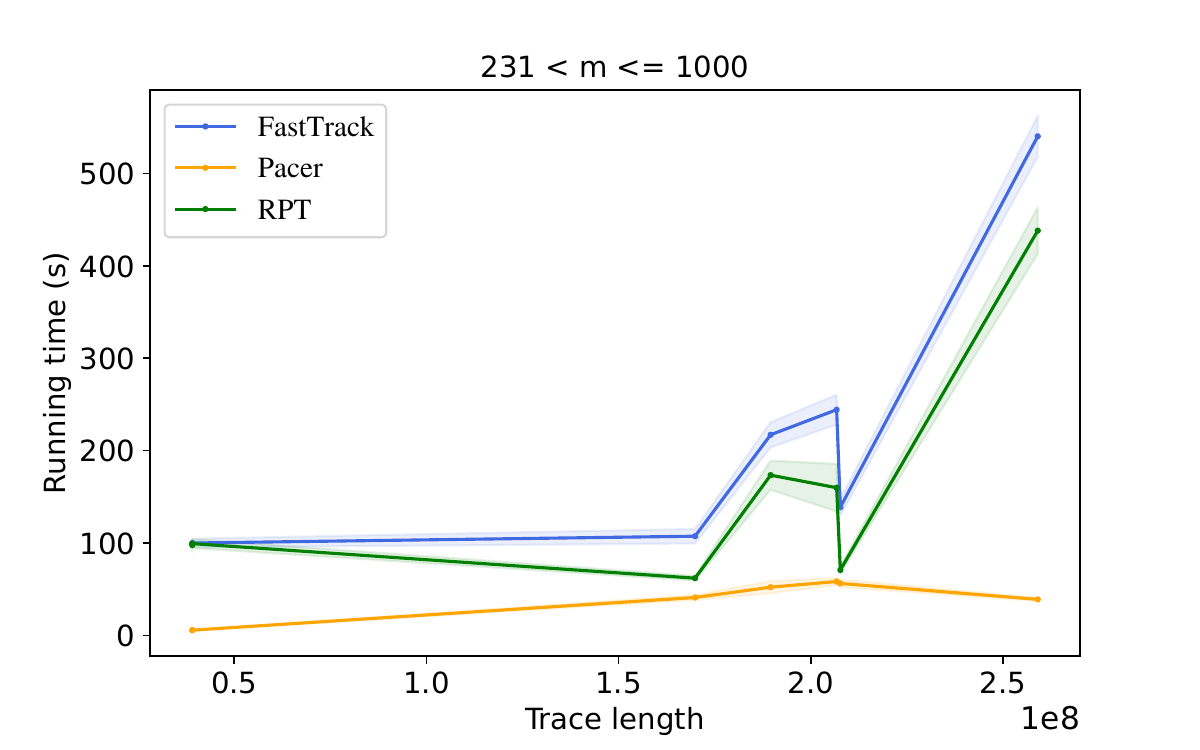}
    \vspace{-0.2in}
    \caption{Running times for traces with $m \in (231, 1000]$}
\end{subfigure}
\caption{Running time as a function for trace length in each cluster.}
\figlabel{clusters-runningtime}
\end{figure*}

We next `zoom-in' into the traces to understand the running times better, instead
of aggregating the running times for several traces.
Indeed, the number of traces in the extremal buckets (for example, for traces with length $>700$M) averaging
the running times smoothes out interesting behavior that we would like to otherwise understand.
For a fine grained analysis, we cluster the traces according to the parameter $m$
and analyze each such cluster individually.
\figref{clusters-runningtime} shows how the running time varies
with trace lengths, in each cluster.
The first observation we make is that the exact runtimes vary widely across clusters
(even for similar trace lengths);
see for example the clusters corresponding to $m \in (29, 59]$ and $m \in (95, 231]$
where the time taken varies significantly for traces of similar lenghts.
However, inside a given cluster, the times increase, roughly linearly with the lengths of the traces.
This justifies our choice of $m$ as a measure for clustering traces.
Indeed, the number of threads (and thus $m$) governs the
size of the vector clocks and also the running time, and further,
also governs the the number of events sampled by \RND.
Finally, the time taken by \RND, in each cluster, is much lower than
\pacer, which is much lower than the deterministic algorithm of \fasttrack
where every event in the trace is analyzed for detecting data races.

Observe that, even though our theoretical analysis of \RND (\secref{algo}) guarantees \emph{constant} running time,
the trend for \RND in any of the figures in \figref{clusters-runningtime} does not completely
`flatten' out.
This is due to the cost introduced by the random number generator
and in detecting when sampling is switched on or off.
In the next subsection, we investigate the running times in
further detail to highlight this.
Overall, \RND introduces much lower analysis cost than \pacer and \fasttrack.}

%% file: constant.tex

\myparagraph{Constant Running Time}{
\begin{figure}[h]
\begin{subfigure}{0.48\textwidth}
    \includegraphics[width=0.97\textwidth]{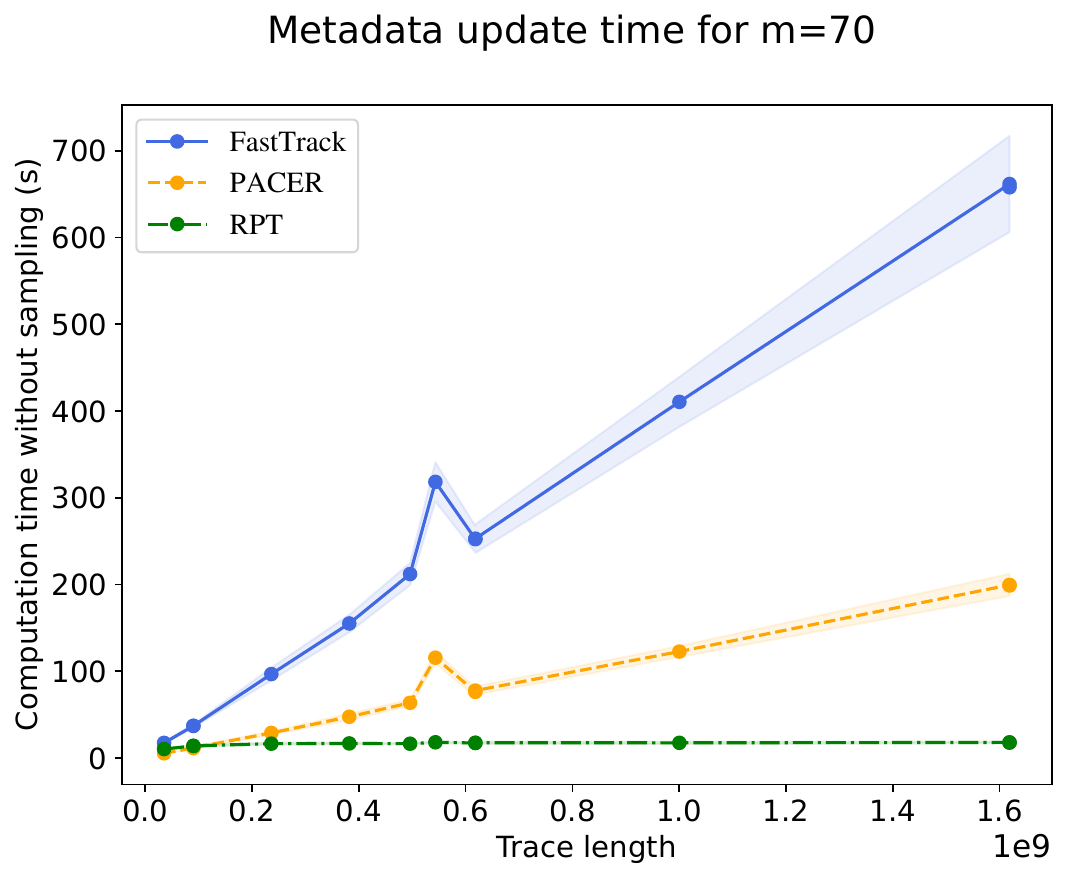}
    \caption{Traces with $m = 70$}
\end{subfigure}
\begin{subfigure}{0.48\textwidth}
    \includegraphics[width=\textwidth]{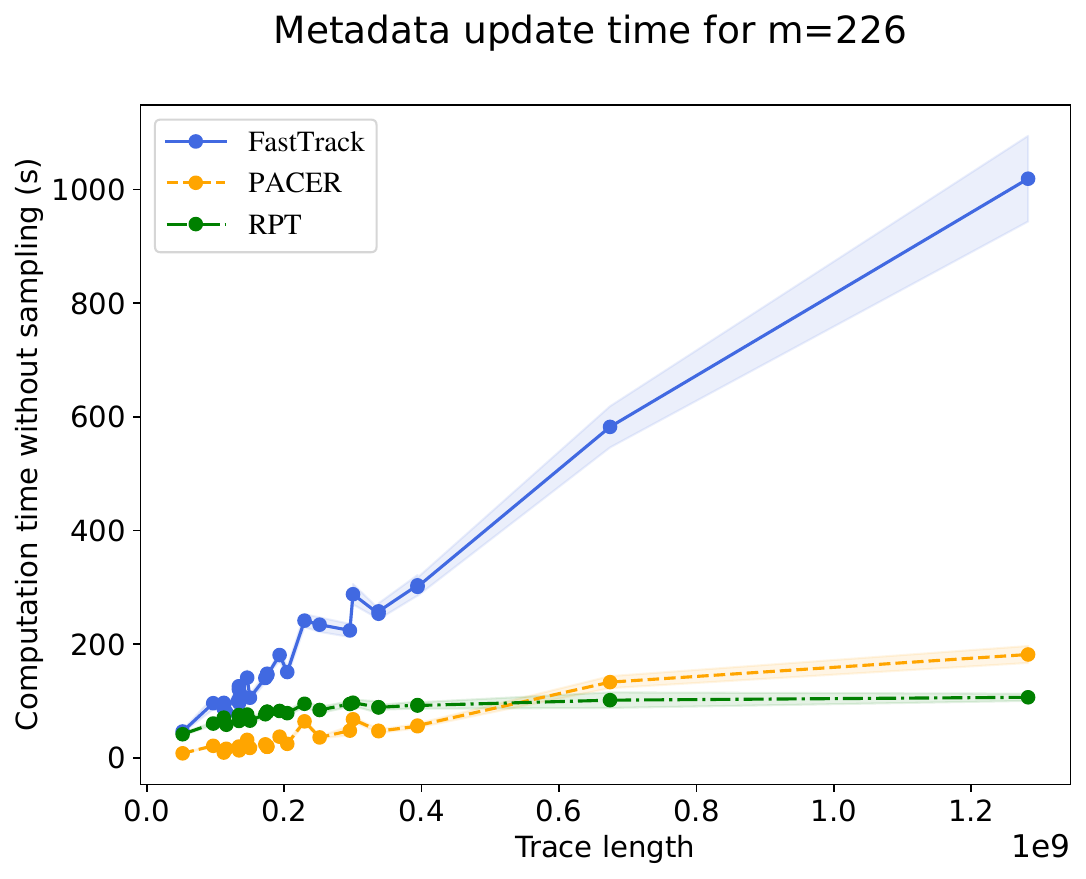}
    \vspace{-0.2in}
    \caption{Traces with $m  = 226$}
\end{subfigure}
\caption{Time to update metadata as a function of trace length, when the parameter $m$ is constant.}
\figlabel{constant}
\end{figure}
Recall that the number of events sampled by \RND is $\widetilde{O}(m)$ and is independent of the length of the trace. A natural question to ask is if that is reflected in \RND's running time experimentally. 
We study two collections of traces from our benchmarks to better understand this. The first set consists of all traces ($12$ in total) with $m = 70$, and the second consists of all traces ($28$ in total) with $m = 226$. Trace lengths vary in each collection to help understand overheads of each algorithm with increasing trace length. 
For both these sets, we plot the overhead due to processing meta-data for 
\fasttrack, \pacer, and \RND for the corresponding set of benchmarks in \figref{constant}. 
This excludes the time taken by the sampling-based algorithms in generating
the random numbers.
Since \fasttrack performs meta-data operations on all events, 
we report the total time taken for it. 
For \RND and \pacer, we only report the time for processing meta-data on the chosen events. 
We see that as expected, \RND spends constant amount of time for analyzing the sampled part of the trace.
On the other hand, both \pacer and \fasttrack spend time that increases with trace length.
}

%% file: precision.tex

\myparagraph{Precision}{
\begin{figure}[t]
\includegraphics[width=0.45\textwidth]{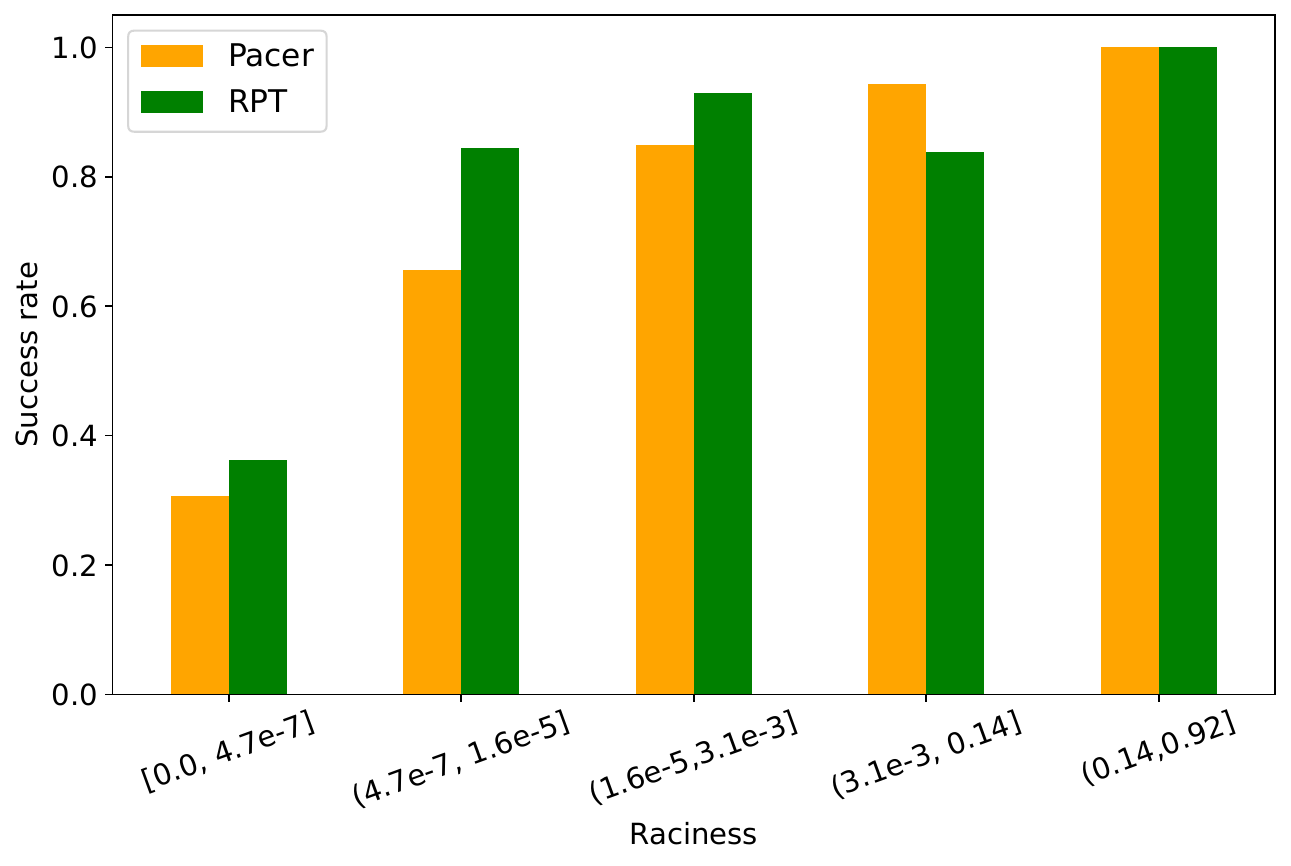}
 \vspace{-0.1in}
\caption{Success rate as a function of raciness (average)}
\label{fig:raciness}
\end{figure}
We want to evaluate
how \RND performs in terms of its ability to expose data races.
Since \RND samples only a constant number
of events, it is bound to not report every dynamic \emph{warning},
i.e., those events that appear as the second event $e_2$ 
in a data race \emph{event} pair $(e_1, e_2)$.
At the same time, the $(\epsilon, \delta)$-guarantee of
\RND ensures that whenever the observed execution is sufficiently `racy',
it will report a race with a high probability.
In the following, we report on the precision of \RND and compare
it with the precision of \pacer.

\figref{raciness} shows the probability with which \pacer and \RND 
detect at least one race for each benchmark; 
we call this the success rate of each benchmark:
\[
	\text{success} = \frac{\text{\# runs with $\geq 1$ warnings}}{\text{Total number of runs}}
\] 
Recall that since \RND is a property tester, it guarantees to detect a race with high probability only when the trace being analyzed is far from race-free traces with respect to hamming distance. This is difficult to measure for a trace. We computed an approximation to the hamming distance that we call the \raciness of a trace:
\[
	\raciness = \frac{\text{avg. \# warnings reported by \fasttrack}}{\text{Trace length}}
\]
Raciness of a trace $\tr$ is an upper bound on the hamming distance of $\tr$ from any race-free trace. In other words, if a trace has low raciness, then it is very close to being race-free with respect to the hamming distance.
However, it could be a poor overestimate. 
We expect that the success rate of \RND will increase as the raciness of the trace increases.
In \figref{raciness},
we cluster traces based on their raciness, and aggregate the success rates for each bucket,
and then evaluate both \RND and \pacer based on their success rates.
Observe that over most of the clusters,  \RND's probability of detecting a race is similar, 
if not better, than \pacer's probability of race detection. 
Overall, we conclude that \RND successfully flags an execution racy with very high probability, 
even if the number of warnings in the trace is small (about $10^{-7}$ times the number of events in the trace). Recall that in our experiments, we run \RND with $\epsilon = 0.01$ and $10^{-7} \ll 0.01 = \epsilon$.
}

%% file: longrace.tex
\myparagraph{Distance between racy pairs}{
\begin{figure}[h]
\includegraphics[width=0.45\textwidth, height=0.3\textwidth]{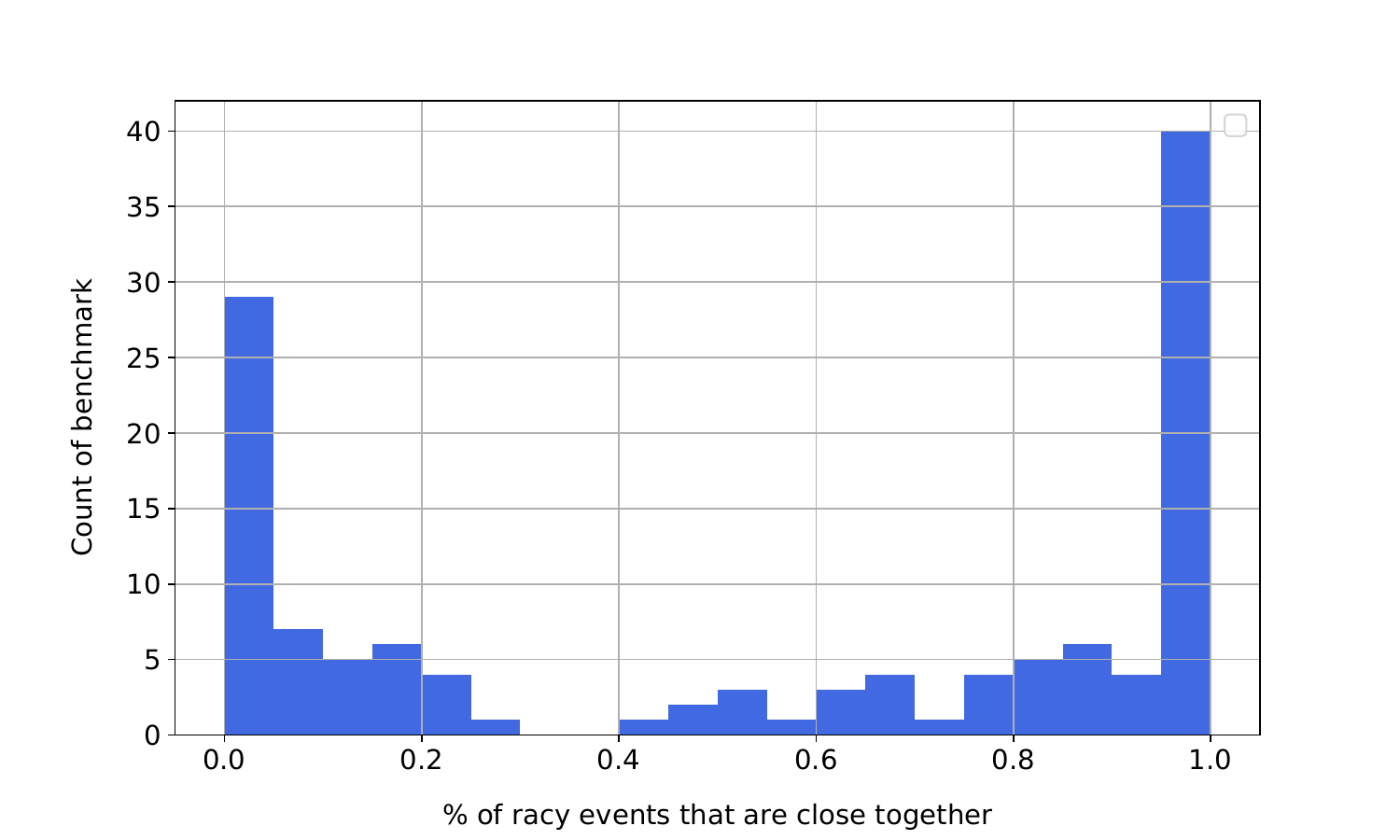}
\vspace{-0.1in}
\caption{Distribution of number of benchmarks over ratios of racy second events over all racy second events.}
\figlabel{short_races}
\end{figure}
\RND samples short sub-traces, and can only report races between pairs of events that both belong to a sampled sub-trace. Thus, \RND cannot report races between pairs of events that are far apart. How often do traces have races separated by a small number of events? Let us say that a race pair $(e_1,e_2)$ in trace $\tr$ is ``short'', if the distance between $e_1$ and $e_2$ is less than the length of the sub-traces sampled by $\RND$. We cluster our benchmark traces based on the number of short races as a fraction of the total number of races in the trace, and this is plotted as a histogram in \figref{short_races}; when plotting this histogram, we drop the one race-free trace in our suite. We observe that in $52$ traces the percentage of short races is $< 25\%$, while in $96$ traces (approximately $2/3$rds of our suite), the percentage of short races is at least $40\%$. Thus, short races are quite common in practice. 

Next, we study how \RND does on traces where the number of short races is very small, i.e., $< 5\%$. There are $29$ such traces (about $1/5$th of our suite). Surprisingly \RND's success rate is very good even on this set. On average \RND reports at least one race $84\%$ of the time on these traces. On $21$ (out of $29$) traces, \RND reports a race $100\%$ of the time, and it reports a race at least half the time on $25$ of these traces. In some of these examples, the percentage of short races is less than $0.01\%$. Among the remaining $4$ examples, there was one trace where \RND never reported a race, and on the remaining $3$ examples whose percentage of short races is $0.003-0.02\%$, \RND reported a race $20-25\%$ of the time. In particular, the one example on which \RND never reports a race, there are no short races and \pacer also fails to report any race in any of its runs on this trace.
}

%% file: observation.tex
\subsection{Detecting Various Races}
\seclabel{multiple-races}

\fasttrack reports every race pair and \pacer guarantees to report every race pair with an equal and non-zero probability. \RND, in contrast, does not provide such strong guarantees with it when comes to reporting an arbitrary race pair. It only promises to report some race on traces that are far from race-free traces. Nonetheless, we would like to understand how many races \RND reports and whether there are races that escape detection with \RND. We report the results of investigation in this section. 

\myparagraph{Average ratios of warnings per run}{
\begin{figure*}[h]
\begin{subfigure}{0.48\textwidth}
 \includegraphics[width=\textwidth]{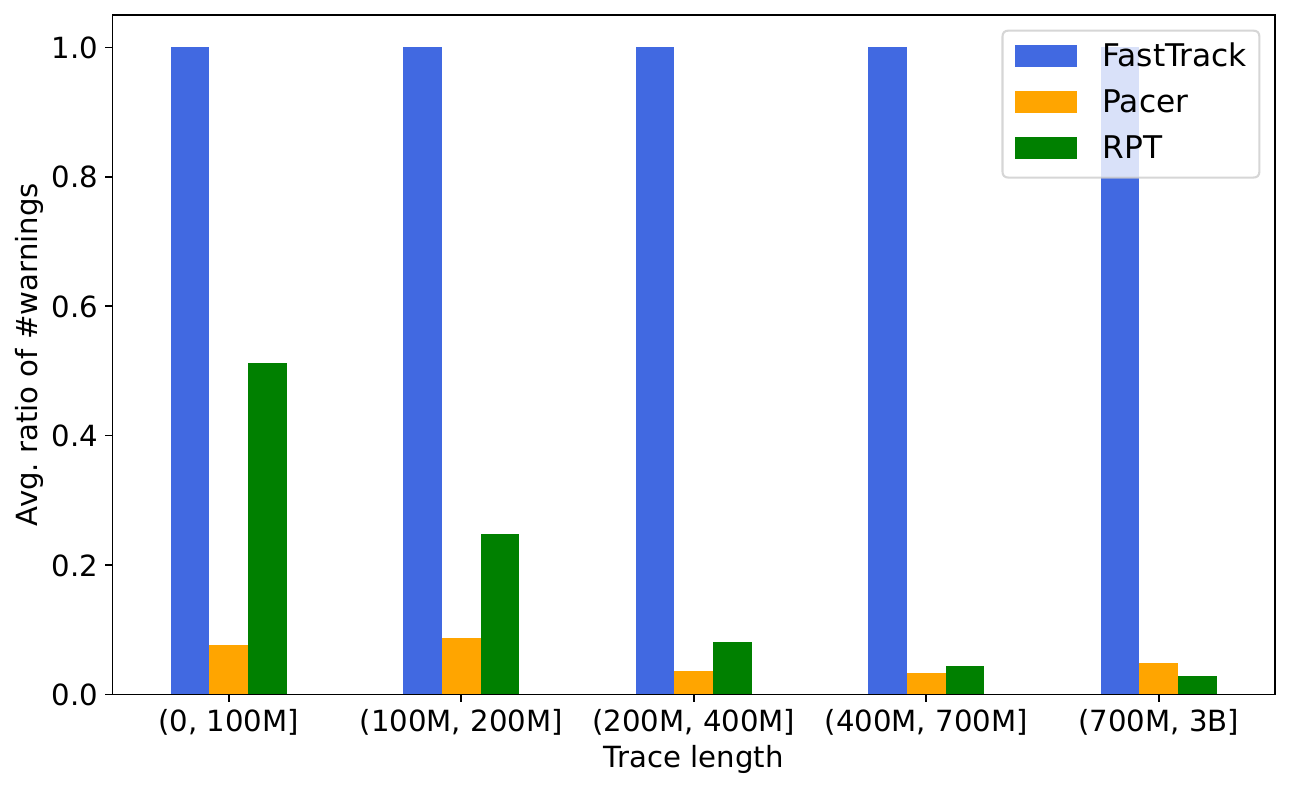}
\caption{Warning ratios v/s trace length}
\figlabel{numWarning-tracelength}
\end{subfigure}
\begin{subfigure}{0.48\textwidth}
 \includegraphics[width=\textwidth]{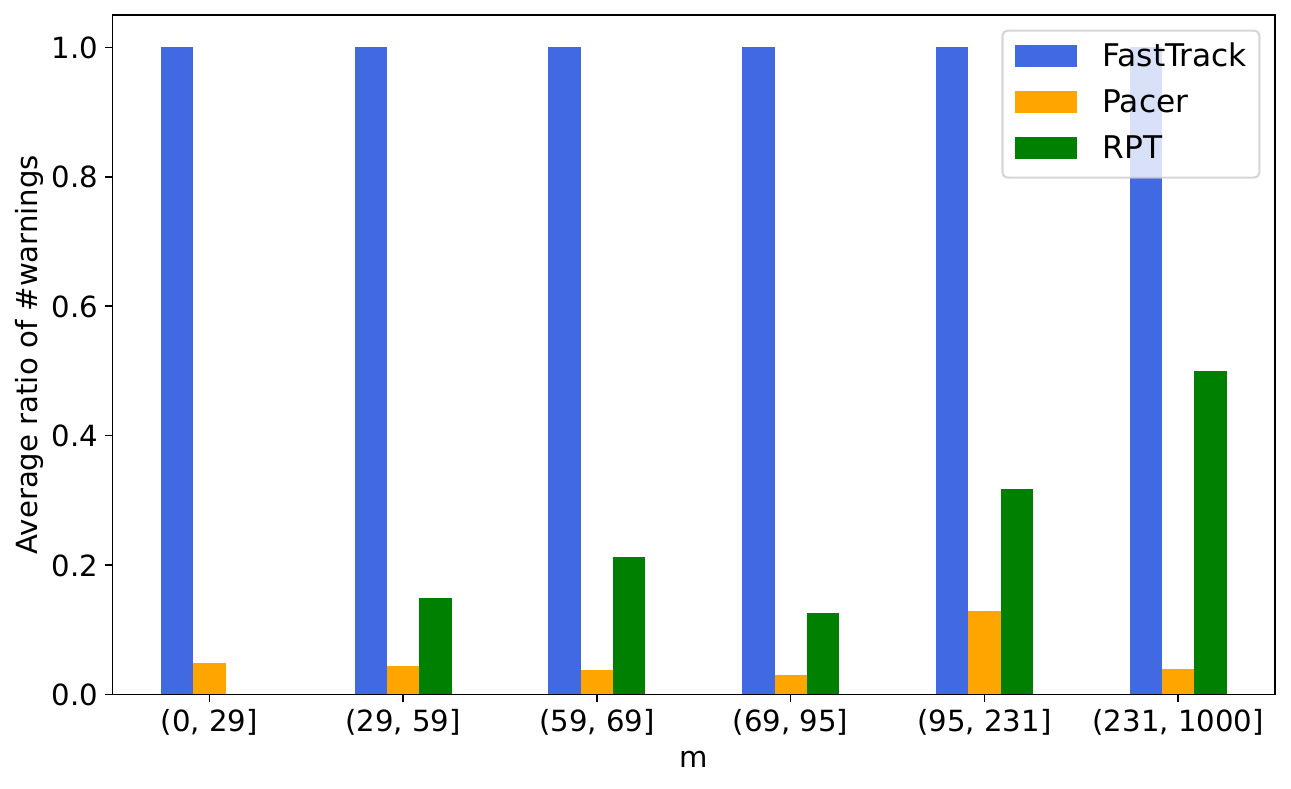}
\caption{Warning ratios v/s $m$}
\figlabel{numWarning-m}
\end{subfigure}
\caption{Ratio of warnings (with warnings of \fasttrack) as a function of trace length and $m$.}
\figlabel{clusters-ratios}
\end{figure*}
In~\figref{clusters-ratios},
we depict how the number of warnings reported varies in our suite of traces.
As before, to visualize the data, we cluster it as per trace length and $m$.
For each cluster, we consider the ratio of warnings reported by
an algorithn (\RND, \pacer or \fasttrack) and the number of warnings
in the trace (namely those reported by the deterministic algorithm \fasttrack).
For each cluster, we compute the average of these ratios.
\figref{numWarning-tracelength}, we report how the ratio varies
across clusters of different trace lengths.
For smaller traces, \RND reports a large fraction of the warnings, as compared to \pacer.
This is expected, because \RND samples constantly many events, which for the case
of smaller traces, amounts to sampling a large fraction of the trace.
As a result, it reports many warnings.
On the other hand, \pacer misses races for smaller traces due to its proportional sampling.
For the large traces, \pacer is able to find more races as expected because of its proportional nature.
\figref{numWarning-m} plots these ratios when the traces are clustered by the value of $m$. The reason to study these plots clustered by $m$ is because the number of samples drawn by \RND on a trace grows as a function of $m$; as shown in Table~\ref{tab:time2}, clustering by $m$ and trace length are different ways to slice up our examples, and there is no correlation between these measures.
}


\begin{figure*}[h]
\begin{subfigure}{0.48\textwidth}
 \includegraphics[width=\textwidth]{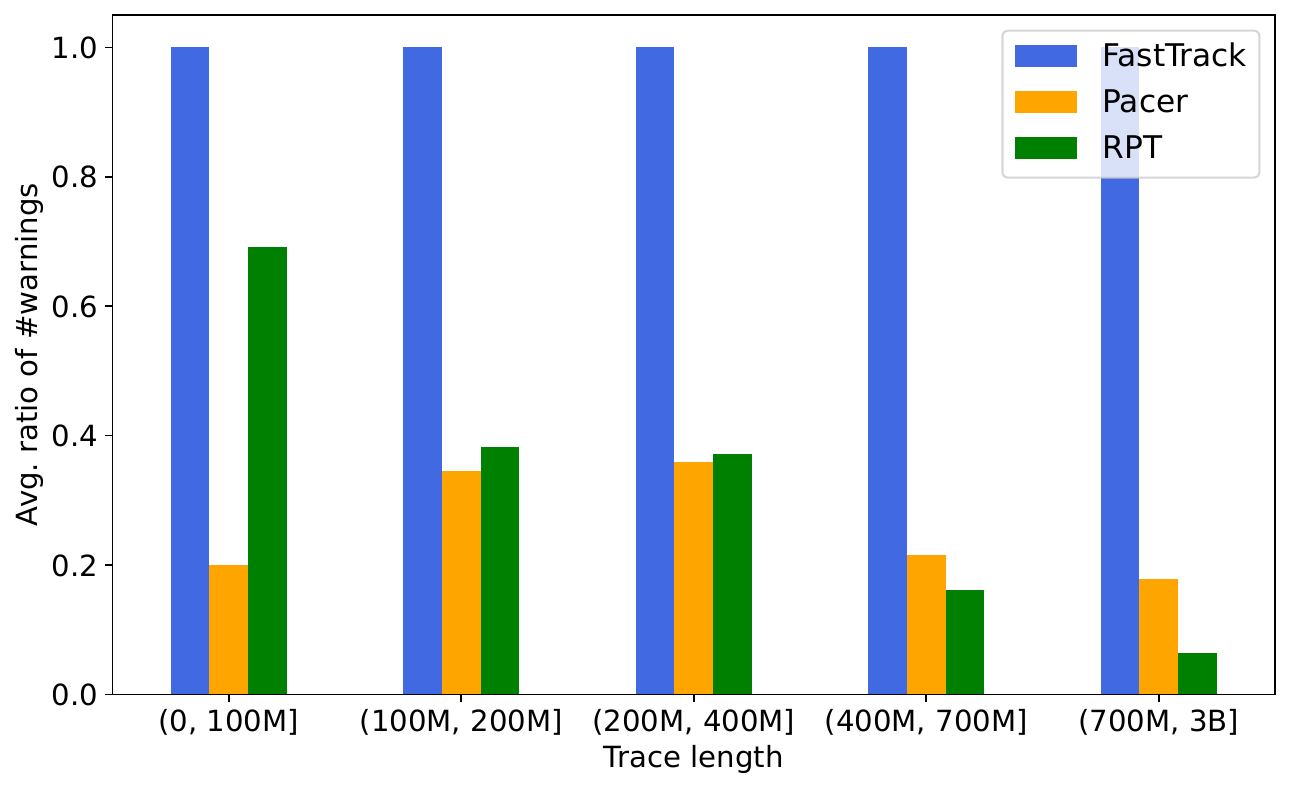}
\caption{Racy memory locations v/s trace length}
\figlabel{numVars-tracelength}
\end{subfigure}
\begin{subfigure}{0.48\textwidth}
 \includegraphics[width=\textwidth]{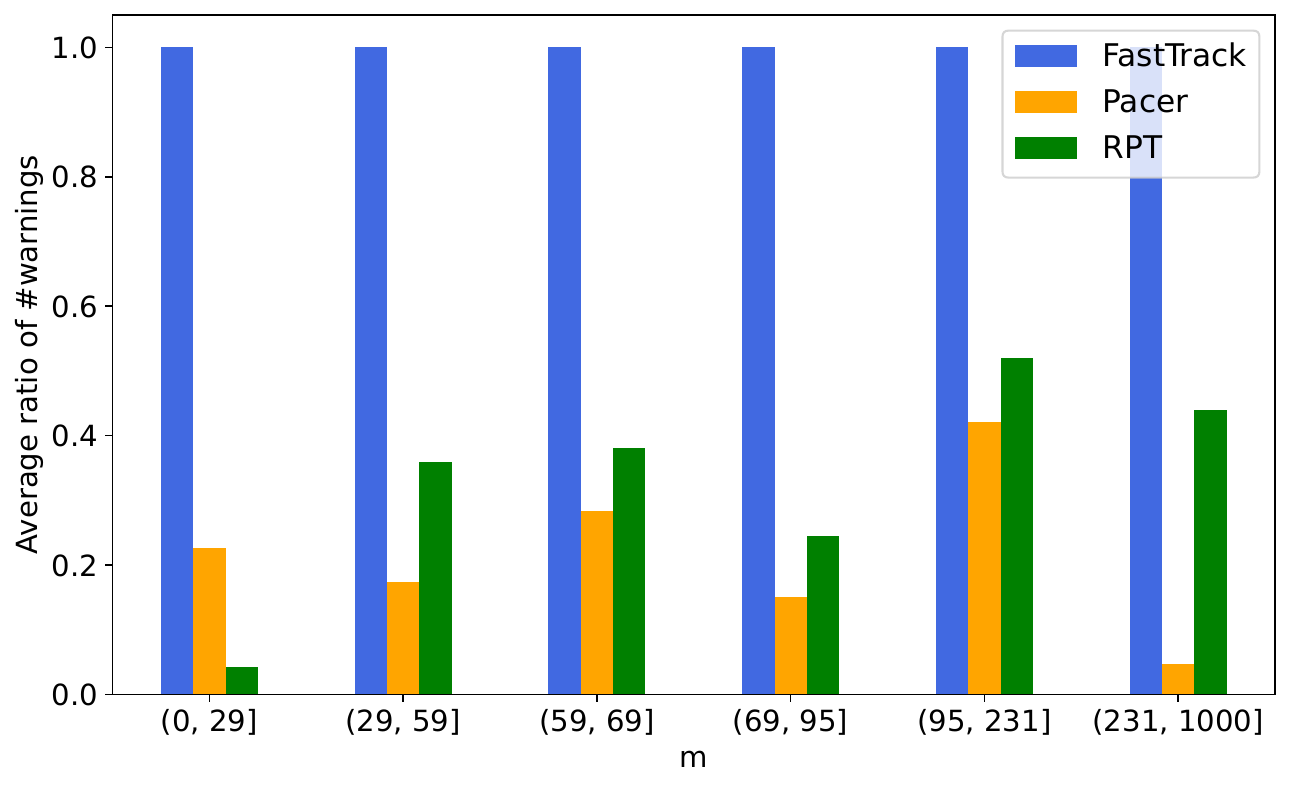}
\caption{Racy memory locations reported v/s $m$}
\figlabel{numVars-m}
\end{subfigure}
\caption{Number of memory locations (normalized by those reported by \fasttrack) as a function of trace length and $m$.}
\figlabel{clusters-ratios-vars}
\end{figure*}

\myparagraph{Exposing racy memory locations}{
Our benchmarks exhibit \acrhb-races on enormous memory locations. 
Here we evaluate the following question --- can \RND detect each racy memory location? 
Or, there are a large number of variables with data races that are inherently difficult for \RND to discover?
Admittedly, reports on unique memory locations are more insightful
for developers using a race detector, as excessive number of
repeated warnings (on the same location) are known to easily
overwhelm developers.
As before, we compute the aggregated ratios of memory locations
that \RND, \pacer and \fasttrack report (as compared with those reported by \fasttrack),
where the aggregation is performed according to trace lengths and $m$.
This is shown in \figref{clusters-ratios-vars}. 
While the trends here are similar to those for number of warnings,
these figures show that, in fact, when we focus on the number of unique memory
locations flagged as racy, \RND is able to correctly flag a good ratio of memory locations to be racy.
}

\begin{figure*}[h]
\begin{subfigure}{0.48\textwidth}
 \includegraphics[width=\textwidth]{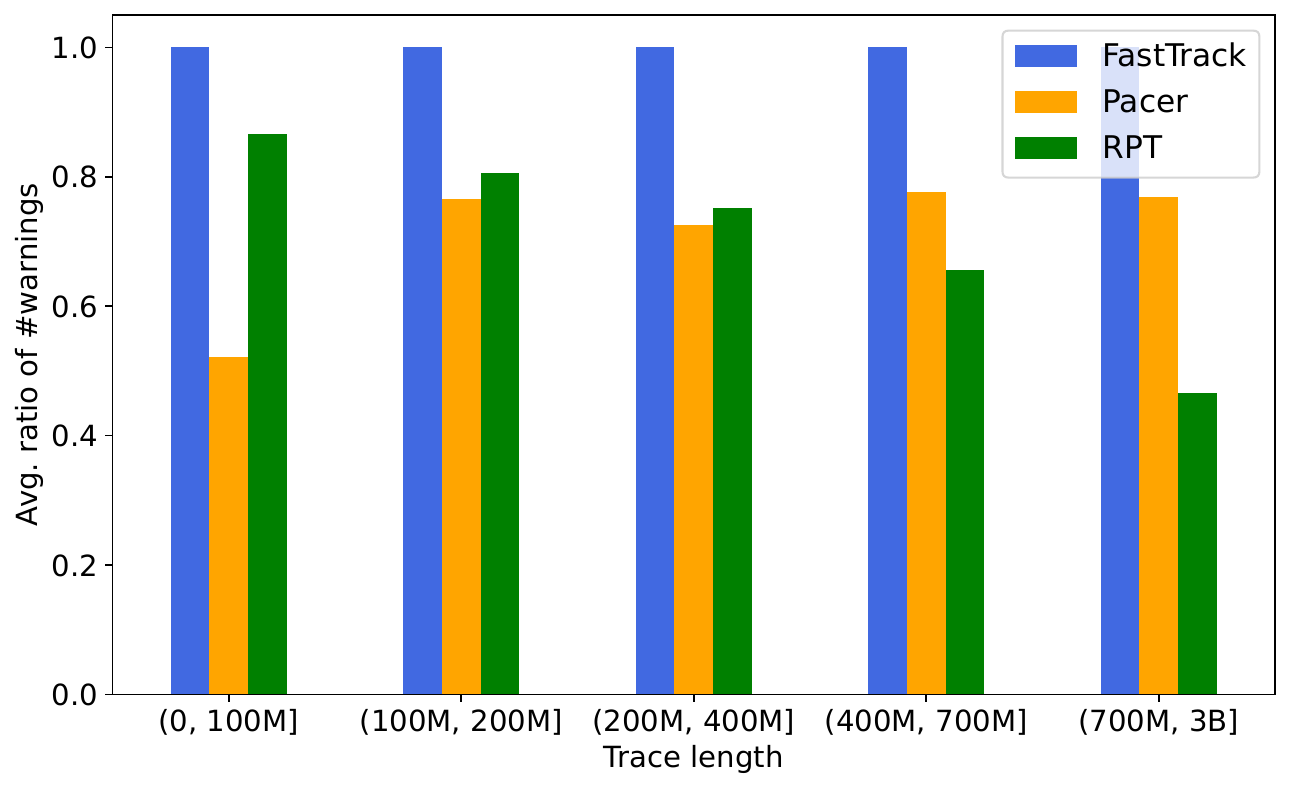}
\caption{Number of racy source code lines v/s trace length}
\figlabel{numLocs-tracelength}
\end{subfigure}
\begin{subfigure}{0.48\textwidth}
 \includegraphics[width=\textwidth]{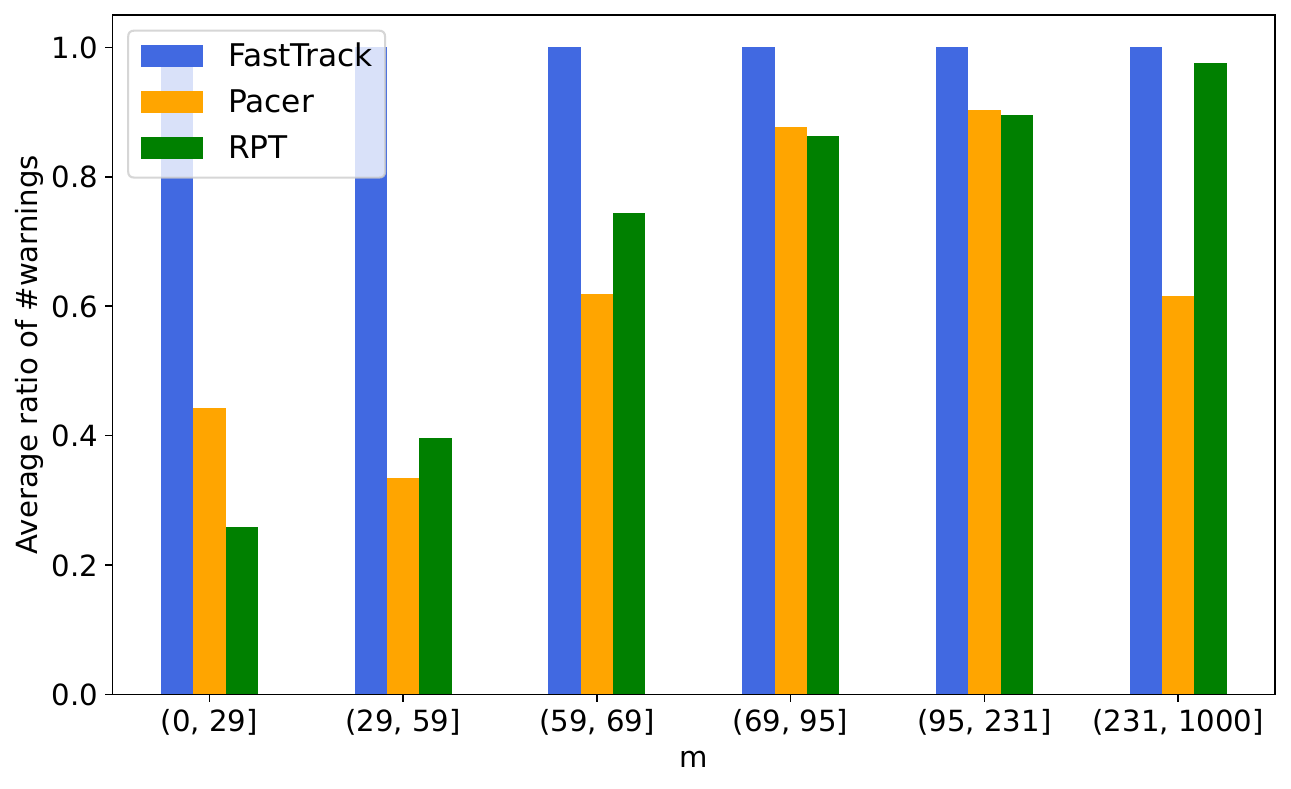}
\caption{Number of racy source code lines v/s $m$}
\figlabel{numSource-m}
\end{subfigure}
\caption{Number of source code locations (normalized by those reported by \fasttrack) as a function of trace length and $m$.}
\figlabel{clusters-ratios-source}
\end{figure*}

\myparagraph{Exposing racy source code locations}{
We next focus on the source code locations that these race detection algorithms report.
From the standpoint of a software developer using a race detector,
this metric is even more valuable since developers are interested in localizing the data races
and thereafter fix them. 
We report the number of source code locations flagged to be racy by each of the three tools,
and present them, as before, by clustering according to trace lengths
and $m$, in \figref{clusters-ratios-source}.
We observe that both \pacer and \RND report reasonably many source code locations.
This shows the power of sampling based approaches.
Notably, for traces with higher $m$ (that is, higher number of threads),
\RND can report almost all locations in the source code that are flagged to be racy 
by the baseline determinsitic algorithm \fasttrack.
}

%% file: epsilon.tex

\subsection{Choosing $\epsilon$}

\begin{figure}[t]
\includegraphics[width=0.45\textwidth, height=0.3\textwidth]{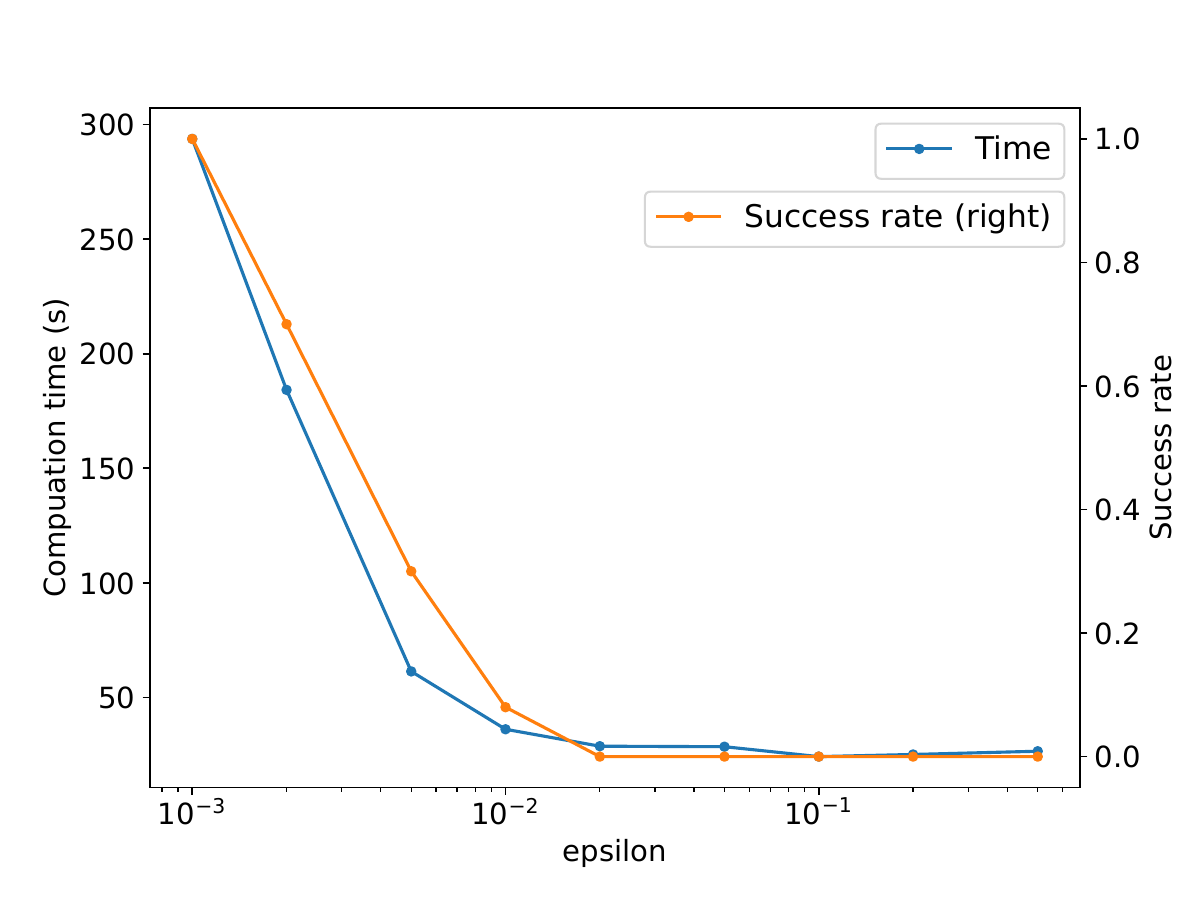}
\vspace{-0.1in}
\caption{The running time and success rate of \RND on \textsf{xalan} as $\epsilon$ varies.}
\figlabel{xalanRateTime}
\end{figure}

The performance of our property testing
algorithm \RND, both in terms of its runtime and its ability to detect races, 
depends on parameters $\epsilon$ and $\delta$. 
Changes to $\delta$ do not significantly affect the number 
of sampled events (which vary as $\ln(\frac{1}{\delta})$). 
Therefore, we keep $\delta$ fixed at $0.1$ and see how things change as we vary $\epsilon$. 
We ran \RND on one of the DaCaPo benchmark \textsf{xalan}, with values of $\epsilon\in\{0.5, 0.2, 0.1, 0.05, 0.02, 0.01, 0.005, 0.002, 0.001\}$.  
The choice of \textsf{xalan} was determined by a desire to pick a long trace with reasonable \raciness. 
In \figref{xalanRateTime}, we plot how running time of \RND changes with $\epsilon$ and how the probability of detecting a race changes with $\epsilon$ (right axis). As expected, both the running time and the probability of detecting a race increase as $\epsilon$ is decreased.

%% file: related.tex

\section{Related Work}
\seclabel{related}

Data races form the most common~\cite{lpsz08}
as well difficult to detect~\cite{heisenbugs} class of concurrency bugs.
Extensive research on developing techniques for automatically
detecting data races has led to the development of many influential
static and dynamic analysis techniques for race detection.
Static techniques such as~\cite{Naik:2006:ESR:1133255.1134018,racerd2018,flanagan2000type,Abadi:2006:TSL:1119479.1119480,voung2007relay}
employ type based analysis or interprocedural data and control flow
analysis to infer if two accesses on a common memory location
may not be protected by a common lock.
Besides the inherent unsoundness in this criteria (memory accesses may still not be racy
even when they are not protected by a common lock), computability
limits imply that static techniques either report many false
races, do not scale or require extensive manual 
annotation in the software to be analyzed.

Dynamic analysis techniques, on the other hand, do not detect races
in the entire source code, but instead
limit their attention to single executions,
and are typically \emph{sound} and completely automated.
Popular techniques include \eraser-style lockset analysis~\cite{savage1997eraser},
happens-before (\acrhb)~\cite{lamport1978time} based race detection~\cite{djit,Pozniansky:2003:EOD:966049.781529,elmas2007goldilocks}
or hybrid techniques~\cite{choi:2003:HDD:781498.781528}.
Amongst these, only \acrhb-based techniques are sound (upto the first race reported~\cite{shb2018}),
and are implemented in popular tools like \tsan~\cite{threadsanitizerLLVM,threadsanitizer},
Helgrind~\cite{helgrind} or Intel Inspector~\cite{intelinspector}
using vector clocks.
Even though such techniques are fully automated, their use is typically
limited to in-house testing to avoid the large overhead of tracking
metadata required for analysis.

\myparagraph{Reducing runtime overhead}{
Our work aims to to improve the performance of \acrhb-race detectors.
The most prominent work with similar goal is the \fasttrack~algorithm
that reduces the overhead of race checking, simplifying it using the \emph{epoch} optimization.
In our evaluation, we compare against this algorithm and show that sampling 
approaches like ours and \pacer's can indeed 
complement its performance, without compromising soundness.
The overhead due to vector clocks updates can also 
be systematically reduced for programs whose execution graphs
exhibit special structures~\cite{cheng1998detecting,Feng1997,surendran2016dynamic,Dimitrov2015,Agrawal2018,raman2012scalable}.
Other works with similar goals include the use of hardware support~\cite{radish2012},
optimizing metadata synchronization~\cite{Wood2017,Bond2013},
or static analysis for optimizing race check placement~\cite{bigfoot,redcard}.
}

\myparagraph{Sampling Based Dynamic Analysis}{
The \pacer~algorithm due to Bond, Coons and McKinley~\cite{bond2010pacer} 
is closest in spirit to our approach, although has significant differences.
First the {\pacer} algorithm has stronger guarantees. The algorithm has a sampling parameter $\rho$ and it guarantees that every race in the input trace will be detected with probability at least $\rho$. A property tester, as presented in this paper, however has weaker guarantees. It promises to detect races in traces that are far, in terms of hamming distance, from race-free executions with high probability. However, in order obtain the stronger guarantees, {\pacer} needs to sample significantly more events than the algorithm presented here. The expected number of samples drawn by {\pacer} is $\rho n$ (where input trace $\tr$ has length $n$), and it can be easily shown that there are executions on which with non-zero probability {\pacer} will sample \emph{all} events in the trace. In contrast, the number of events sampled by our algorithm is at most $\widetilde{O}(t+\lh)$, which is independent of the length of the input trace $\tr$. This means that our algorithm will typically have lower overheads than {\pacer}, at the cost of finding fewer races. This is borne out by our experiments. 
Other sampling based approaches that are closely related include 
LiteRace~\cite{marino2009literace} also resorts to sampling for reducing runtime overhead
guided by the `cold-region hypothesis' that data races are more likely 
exhibited in parts of the source code that are not exercised frequently during the execution.
DataCollider~\cite{datacollider2010}
samples memory locations to detect data races on, in order to reduce runtime overhead of dynamic race detection.
However, most of these algorithms provide no formal mathematical guarantees
for the likelihood of reporting races or any upper bound on the number of events sampled.
}

\myparagraph{Sampling Based Hybrid Analysis}{
The work in RaceMob~\cite{racemob} deploys a two-phase hybrid analysis.
The first static analysis phase identifies a more precise set of memory locations
that may potentially be racy, and in the second pass,
performs sampled dynamic analysis while focusing on the memory locations
identified from during the first static analysis phase.
This is in line with the work of Choi et al~\cite{Choi02}, and is also
more recently deployed for kernel race bugs~\cite{razzer2019}.
We remark that the fundamental algorithmic improvements proposed in our paper are 
orthogonal to such analyses and can be modularly plugged into such hybrid race detection techniques 
to reduce the overhead of the dynamic analysis phase.
Similar to RaceMob's analysis, \RND~can also benefit from
an additional static analysis phase which can identify a focused set
of memory locations.
Such a set can help reduce space usage due to redundant metadata and the performance
cost due to updates on such metadata.
Further, if a candidate set of memory locations is known
prior to the dynamic analysis, the core algorithm
of \RND~can also be adopted so that it only samples sub-traces
that begin at access events of memory locations belonging to this candidate set.
In our paper, we focus on experimental comparison with purely dynamic analysis techniques to
precisely distill the algorithmic benefits our approach offers.
}

\myparagraph{Randomized scheduling}{
	Another popular set of race detection techniques
	that employ randomization include those that
	drive the thread schedulers using 
	randomization~\cite{PCT2010,Sen:2008:RDR:1375581.1375584,c11tester,krace2020},
	with the goal of enhancing the likelihood of covering a racy location.
	Such techniques fall into the general class
	of controlled concurrency testing techniques~\cite{ControlledConcurrency2016}
	that employ heuristics such as bounding context switches~\cite{ContextBounding2007,Kahlon05,penelope2010},
	thread speed control~\cite{Chen2018} or reinforcement learning to determine
	a buggy thread schedule~\cite{QL2020}.
	Our approach is orthogonal to all such approaches
	and can potentially reduce the overhead of intensive exploration.
}

\myparagraph{Other Approaches}{
	Dynamic race detection has been a topic of interest, and
	recent advances such as \emph{predictive} analysis
	aim to enhance coverage of \acrhb-based race detectors
	by considering alternate reorderings that exhibit a race,
	without explicitly re-executing the program, and have been 
	shown to be effective in catching hidden races.
	These include explicit enumeration based techniques~\cite{sen2005detecting},
	symbolic techniques~\cite{Said2011,rv2014},
	graph based analyses~\cite{PavlogiannisPOPL20}
	and partial order based techniques~\cite{cp2012,wcp2017,SyncP2021,Roemer20,Roemer18}.
}

\myparagraph{Property Testing}{
Property testing~\cite{prop-test-book} is a widely studied sub-field in theoretical computer science, where the relaxation to solve a decision problem is exploited to design sublinear algorithms for a variety of problems. Our algorithm and its proof, draws heavily from the ideas presented in~\cite{akns99}, where it is shown that regular languages can be property tested with a constant number of samples. The proof we present here specializes the ideas in~\cite{akns99} for race detection. In particular, we have paid particular attention to the constants involved in our analysis, because even small factor changes can impact the performance of our algorithm on benchmarks we experiment with.
}

%% file: conclusions.tex

\section{Conclusions and Future Work}
\seclabel{conclusions}

We presented a randomized property tester (\RND) for detecting {\acrhb}-races. The algorithm is sound, i.e., never reports a race when there is none, and guarantees to detect a race if the input trace is far from any race-free trace. Moreover, the algorithm samples and processes only constantly many events even in the worst case. This is in sharp contrast to previously proposed sampling based approaches for race detection. \RND was implemented and compared against two well known {\acrhb} race detectors: {\fasttrack} and {\pacer}. Experimental evaluation showed that {\RND} did indeed have the lowest running time among all the algorithms, and detects races often.

There are several interesting avenues for future work.
Our work demonstrates the effectiveness of property testing based algorithms for detecting data races. 
We think that this paradigm will also be promising for improving the effectiveness of detecting other errors such as deadlocks~\cite{havelund2000} and atomicity violations~\cite{MathurAtomicity20,Flanagantypesatomicity2008,Farzan2009cav} for which dynamic analysis is often the preferred method of detection.
Second, it would be interesting to design property testing algorithms for more predictive notions of data races including polynomial time algorithms~\cite{cp2012,wcp2017,SyncP2021} as well as more exhaustive ones~\cite{rv2014,Said2011,Mathur20}.
Finally, we envision that developing a sampling algorithm (like \RND or \pacer) that performs sampling in a distributed manner, and introduces minimal additional synchronizations on events will help further lower the overhead introduced due to sampling, but is expected to require new algorithmic insights.

%% file: acks.tex
\begin{acks}
We thank anonymous reviewers for their constructive feedback on
an earlier drafts of this manuscript. Umang Mathur was partially
supported by the Simons Institute for the Theory of Computing, and by a Singapore Ministry of Education (MoE) Academic Research Fund (AcRF) Tier 1 grant.
Minjian Zhang and Mahesh Viswanathan are partially supported by NSF SHF 1901069 and NSF CCF 2007428. 
\end{acks}

%% file: tables/table1.tex

\begin{table*}[tbh]
\label{table:table1}

\centering
\scalebox{0.63}{
\begin{adjustbox}{center}
\renewcommand{\arraystretch}{1.0}
\begin{tabular*}{1.60\columnwidth}{!{\VRule[1pt]}c!{\VRule[1pt]}c|c!{\VRule[1pt]}c|c!{\VRule[1pt]}c|c|c!{\VRule[1pt]}c|c|c!{\VRule[1pt]}}
\specialrule{1pt}{0pt}{0pt}
1 & 2 & 3 & 4 & 5 & 6 & 7 & 8 & 9& 10 & 11\\
\specialrule{1pt}{0pt}{0pt}
\rowcolor[HTML]{EFEFEF} 
\cellcolor[HTML]{EFEFEF} Benchmark
& \cellcolor[HTML]{EFEFEF} {trace length}
& \cellcolor[HTML]{EFEFEF} $M$
& \cellcolor[HTML]{EFEFEF} {\pacer sampled}
& \cellcolor[HTML]{EFEFEF} {\RND sampled}
& \multicolumn{3}{c!{\VRule[1pt]}}{\cellcolor[HTML]{EFEFEF}{\# Warnings}}

& \multicolumn{3}{c!{\VRule[1pt]}}{\cellcolor[HTML]{EFEFEF}{\# Warning variables }} \\
\cmidrule[0.5pt]{6-11}

\rowcolor[HTML]{EFEFEF} 
 \cellcolor[HTML]{EFEFEF} 
& \cellcolor[HTML]{EFEFEF} 
& \cellcolor[HTML]{EFEFEF} 
& \cellcolor[HTML]{EFEFEF} 
& \cellcolor[HTML]{EFEFEF} 
& \cellcolor[HTML]{EFEFEF} \textsc{FT}
& \cellcolor[HTML]{EFEFEF} \pacer
& \cellcolor[HTML]{EFEFEF} \RND
& \cellcolor[HTML]{EFEFEF} \textsc{FT}
& \cellcolor[HTML]{EFEFEF} \pacer
& \cellcolor[HTML]{EFEFEF} \RND\\

\specialrule{1pt}{0pt}{0pt}
\textsf{zero-reversal-logs-final-logs} & 1.0M & 56.0 & 358.9K & 499.0K & 29.6K & 3.5K & 18.4K & 988.00 & 163.53 & 634.67\\
\textsf{zero-reversal-logs-final-logs1} & 1.4M & 32.0 & 136.5K & 814.5K & 72.00 & 4.58 & 39.63 & 24.00 & 1.22 & 12.83\\
\textsf{zero-reversal-logs-final-logs2} & 3.1M & 60.0 & 747.3K & 2.6M & 14.00 & 0.29 & 8.80 & 13.00 & 0.27 & 8.20\\
\textsf{zero-reversal-logs-final-logs3} & 11.7M & 72.0 & 1.5M & 11.1M & 252.00 & 7.13 & 123.98 & 28.00 & 1.02 & 17.05\\
\textsf{HPCBench-NPBS-DC.S-12M-events} & 11.7M & 228.0 & 2.9M & 11.3M & 5.8K & 119.09 & 5.4K & 574.00 & 17.62 & 534.35\\
\textsf{HPCBench-NPBS-DC.S-12M-events1} & 11.7M & 68.0 & 3.5M & 11.1M & 82.00 & 2.84 & 18.17 & 82.00 & 2.84 & 18.17\\
\textsf{sunflow} & 16.8M & 68.0 & 1.3M & 15.1M & 252.00 & 2.35 & 109.96 & 28.00 & 0.24 & 15.90\\
\textsf{misc-hsqldb-hsqldb} & 18.8M & 184.0 & 3.0M & 18.1M & 284.00 & 5.40 & 258.87 & 5.00 & 0.75 & 4.55\\
\textsf{DRACC-DRACC-OMP-017-Counter-wr} & 27.0M & 68.0 & 11.6M & 21.8M & 15.00 & 1.09 & 13.12 & 15.00 & 1.09 & 13.12\\
\textsf{OmpSCR-v2.0-c-testPath-30M-eve} & 30.2M & 68.0 & 18.0M & 23.4M & 181.00 & 12.89 & 26.37 & 16.00 & 0.89 & 11.73\\
\textsf{OMPRacer-Lulesh-35M-events-16} & 35.3M & 70.0 & 8.0M & 26.0M & 8.8M & 253.2K & 1.9M & 65.4K & 36.5K & 51.4K\\
\textsf{OmpSCR-v2.0-c-testPath-37M-eve} & 37.5M & 228.0 & 19.8M & 35.8M & 289.00 & 8.80 & 207.83 & 57.00 & 2.05 & 52.48\\
\textsf{misc-tradesoap-tradesoap} & 39.1M & 904.0 & 4.8M & 38.8M & 7.4K & 247.87 & 7.2K & 396.00 & 9.60 & 387.88\\
\textsf{misc-tradebeans-tradebeans} & 39.1M & 908.0 & 4.6M & 39.0M & 7.2K & 22.09 & 7.1K & 401.00 & 7.93 & 399.45\\
\textsf{series} & 40.0M & 18.0 & 22.3K & 10.3M & 0.00 & 0.00 & 0.00 & 0.00 & 0.00 & 0.00\\
\textsf{DataRaceBench-DRB155-missingor} & 50.0M & 68.0 & 20.6M & 29.7M & 17.00 & 0.04 & 7.60 & 17.00 & 0.04 & 7.60\\
\textsf{DataRaceBench-DRB155-missingor1} & 50.0M & 228.0 & 27.3M & 46.5M & 58.00 & 55.11 & 38.85 & 58.00 & 55.11 & 38.85\\
\textsf{OMPRacer-Lulesh-52M-events-56} & 52.1M & 226.0 & 11.0M & 49.4M & 12.1M & 379.2K & 6.8M & 86.4K & 48.6K & 86.3K\\
\textsf{zero-reversal-logs-final-logs4} & 58.5M & 42.0 & 8.8M & 22.8M & 93.00 & 1.71 & 0.98 & 4.00 & 0.20 & 0.32\\
\textsf{tomcat} & 63.2M & 212.0 & 4.8M & 56.6M & 1.2M & 35.4K & 498.9K & 17.8K & 9.2K & 16.4K\\
\textsf{OmpSCR-v2.0-cpp-sortOpenMP-cpp} & 88.9M & 66.0 & 35.9M & 35.0M & 32.1M & 986.0K & 240.7K & 8.0K & 5.4K & 7.9K\\
\textsf{DRB177-fib-taskdep-yes-90M-eve} & 90.2M & 70.0 & 44.3M & 37.2M & 3.8K & 112.52 & 1.2K & 1.5K & 102.50 & 692.37\\
\textsf{DataRaceBench-DRB176-fib-taskd} & 90.2M & 70.0 & 44.3M & 37.0M & 9.9K & 334.80 & 3.3K & 2.3K & 256.53 & 1.4K\\
\textsf{DRB177-fib-taskdep-yes-90M-eve1} & 90.3M & 230.0 & 44.2M & 74.8M & 7.8K & 452.84 & 4.8K & 3.9K & 425.11 & 2.7K\\
\textsf{DataRaceBench-DRB176-fib-taskd1} & 90.3M & 230.0 & 44.3M & 74.6M & 26.6K & 934.31 & 19.5K & 7.2K & 783.18 & 6.3K\\
\textsf{fop} & 96.0M & 8.0 & 9.9M & 5.4M & 0.00 & 0.00 & 0.00 & 0.00 & 0.00 & 0.00\\
\textsf{OmpSCR-v2.0-c-LoopsWithDepende} & 96.4M & 66.0 & 6.7M & 36.1M & 1.7K & 113.79 & 551.70 & 31.00 & 14.43 & 21.02\\
\textsf{OmpSCR-v2.0-c-LoopsWithDepende1} & 96.4M & 66.0 & 6.8M & 35.9M & 2.9K & 97.56 & 569.07 & 31.00 & 14.55 & 20.35\\
\textsf{OMPRacer-XSBench-97M-events-16} & 96.6M & 66.0 & 23.0M & 36.4M & 27.00 & 5.09 & 7.45 & 27.00 & 5.09 & 7.45\\
\textsf{OMPRacer-XSBench-97M-events-56} & 96.6M & 226.0 & 21.4M & 77.5M & 89.00 & 6.69 & 61.02 & 89.00 & 6.69 & 61.02\\
\textsf{DRACC-DRACC-OMP-014-Counter-wr} & 104.9M & 68.0 & 54.2M & 37.5M & 15.00 & 0.27 & 6.87 & 15.00 & 0.27 & 6.87\\
\textsf{DRACC-DRACC-OMP-020-Counter-wr} & 104.9M & 68.0 & 52.9M & 37.9M & 15.00 & 0.27 & 5.28 & 15.00 & 0.27 & 5.28\\
\textsf{DRACC-DRACC-OMP-019-Counter-wr} & 104.9M & 68.0 & 3.9M & 37.3M & 94.6M & 2.8M & 33.2M & 527.00 & 512.27 & 515.25\\
\textsf{DRACC-DRACC-OMP-018-Counter-wr} & 104.9M & 68.0 & 3.9M & 37.8M & 93.9M & 2.8M & 33.3M & 527.00 & 512.55 & 518.25\\
\textsf{DRACC-DRACC-OMP-013-Counter-wr} & 104.9M & 68.0 & 3.9M & 36.6M & 96.0M & 2.9M & 33.0M & 527.00 & 512.00 & 517.40\\
\textsf{DRACC-DRACC-OMP-012-Counter-wr} & 104.9M & 68.0 & 3.9M & 37.8M & 95.6M & 2.8M & 33.9M & 527.00 & 512.27 & 517.25\\
\textsf{OmpSCR-v2.0-cpp-sortOpenMP-cpp1} & 106.7M & 228.0 & 74.1M & 82.3M & 193.00 & 7.65 & 146.87 & 56.00 & 3.91 & 44.08\\
\textsf{OmpSCR-v2.0-cpp-sortOpenMP-cpp2} & 106.8M & 228.0 & 74.2M & 82.3M & 177.00 & 5.38 & 140.02 & 56.00 & 1.95 & 48.65\\
\textsf{OmpSCR-v2.0-cpp-sortOpenMP-cpp3} & 107.0M & 68.0 & 41.7M & 38.0M & 47.00 & 1.85 & 14.50 & 16.00 & 0.89 & 5.13\\
\textsf{OmpSCR-v2.0-cpp-sortOpenMP-cpp4} & 107.1M & 228.0 & 74.5M & 82.4M & 199.00 & 5.98 & 150.23 & 56.00 & 0.96 & 42.52\\
\textsf{OmpSCR-v2.0-cpp-sortOpenMP-cpp5} & 107.5M & 68.0 & 42.3M & 38.0M & 48.00 & 1.38 & 15.97 & 16.00 & 0.60 & 6.28\\
\textsf{DataRaceBench-DRB154-missinglo} & 112.0M & 68.0 & 56.4M & 38.3M & 15.00 & 0.27 & 5.78 & 15.00 & 0.27 & 5.78\\
\textsf{DataRaceBench-DRB152-missinglo} & 112.0M & 68.0 & 34.4M & 37.8M & 16.00 & 0.09 & 5.00 & 16.00 & 0.09 & 5.00\\
\textsf{DataRaceBench-DRB150-missinglo} & 112.0M & 68.0 & 34.3M & 38.3M & 16.00 & 0.05 & 6.00 & 16.00 & 0.05 & 6.00\\
\textsf{DataRaceBench-DRB122-taskundef} & 112.0M & 66.0 & 61.6M & 37.4M & 15.00 & 0.82 & 4.00 & 15.00 & 0.82 & 4.00\\
\textsf{DataRaceBench-DRB123-taskundef} & 112.0M & 66.0 & 31.6M & 37.4M & 14.0M & 409.2K & 4.7M & 227.00 & 212.00 & 217.75\\
\textsf{DataRaceBench-DRB122-taskundef1} & 112.0M & 226.0 & 61.2M & 84.2M & 55.00 & 55.00 & 24.02 & 55.00 & 55.00 & 24.02\\
\textsf{DataRaceBench-DRB123-taskundef1} & 112.0M & 226.0 & 36.6M & 84.0M & 13.9M & 406.6K & 10.4M & 712.00 & 712.00 & 681.83\\
\textsf{OMPRacer-Kripke-119M-events-56} & 119.2M & 228.0 & 44.6M & 87.3M & 10.4M & 999.4K & 640.1K & 116.3K & 31.4K & 49.6K\\
\textsf{DataRaceBench-DRB110-ordered-o} & 120.0M & 68.0 & 55.4M & 38.7M & 15.00 & 0.27 & 4.25 & 15.00 & 0.27 & 4.25\\
\specialrule{1pt}{0pt}{0pt}
\end{tabular*}
\end{adjustbox}
}
\label{tab:time1}
\end{table*}

%% file: tables/table3.tex

\begin{table*}[t!]
\label{table:table3}

\centering
\scalebox{0.63}{
\begin{adjustbox}{center}
\renewcommand{\arraystretch}{1.0}
\begin{tabular*}{1.61\columnwidth}{!{\VRule[1pt]}c!{\VRule[1pt]}c|c!{\VRule[1pt]}c|c!{\VRule[1pt]}c|c|c!{\VRule[1pt]}c|c|c!{\VRule[1pt]}}
\specialrule{1pt}{0pt}{0pt}
1 & 2 & 3 & 4 & 5 & 6 & 7 & 8 & 9& 10 & 11\\
\specialrule{1pt}{0pt}{0pt}
\rowcolor[HTML]{EFEFEF} 
\cellcolor[HTML]{EFEFEF} Benchmark
& \cellcolor[HTML]{EFEFEF} {trace length}
& \cellcolor[HTML]{EFEFEF} $M$
& \cellcolor[HTML]{EFEFEF} {\pacer sampled}
& \cellcolor[HTML]{EFEFEF}{\RND sampled}
& \multicolumn{3}{c!{\VRule[1pt]}}{\cellcolor[HTML]{EFEFEF}{\# Warnings}}

& \multicolumn{3}{c!{\VRule[1pt]}}{\cellcolor[HTML]{EFEFEF}{\# Warning variables }} \\
\cmidrule[0.5pt]{6-11}

\rowcolor[HTML]{EFEFEF} 
 \cellcolor[HTML]{EFEFEF} 
& \cellcolor[HTML]{EFEFEF} 
& \cellcolor[HTML]{EFEFEF} 
& \cellcolor[HTML]{EFEFEF} 
& \cellcolor[HTML]{EFEFEF} 
& \cellcolor[HTML]{EFEFEF} \textsc{FT}
& \cellcolor[HTML]{EFEFEF} \pacer
& \cellcolor[HTML]{EFEFEF} \RND
& \cellcolor[HTML]{EFEFEF} \textsc{FT}
& \cellcolor[HTML]{EFEFEF} \pacer
& \cellcolor[HTML]{EFEFEF} \RND\\

\specialrule{1pt}{0pt}{0pt}

\textsf{DataRaceBench-DRB110-ordered-o1} & 120.0M & 228.0 & 68.9M & 87.2M & 55.00 & 55.00 & 25.90 & 55.00 & 55.00 & 25.90\\
\textsf{zero-reversal-logs-final-logs5} & 122.5M & 36.0 & 27.5M & 22.1M & 77.00 & 3.69 & 11.98 & 34.00 & 1.73 & 5.47\\
\textsf{crypt} & 126.0M & 30.0 & 105.8M & 18.9M & 0.00 & 0.00 & 0.00 & 0.00 & 0.00 & 0.00\\
\textsf{OMPRacer-QuickSilver-133M-even} & 132.6M & 228.0 & 47.1M & 90.6M & 1.1M & 33.0K & 344.2K & 121.3K & 9.1K & 11.3K\\
\textsf{DataRaceBench-DRB105-taskwait} & 134.0M & 66.0 & 45.2M & 38.6M & 6.6K & 209.44 & 1.6K & 1.3K & 172.93 & 719.35\\
\textsf{DataRaceBench-DRB106-taskwaitm} & 134.0M & 66.0 & 45.5M & 38.1M & 1.3K & 36.47 & 176.00 & 877.00 & 35.82 & 160.47\\
\textsf{DataRaceBench-DRB105-taskwait1} & 134.0M & 226.0 & 41.4M & 91.5M & 40.9K & 1.5K & 23.8K & 4.6K & 1.1K & 4.2K\\
\textsf{DataRaceBench-DRB106-taskwaitm1} & 134.0M & 226.0 & 43.5M & 90.7M & 3.3K & 278.91 & 1.3K & 2.5K & 275.85 & 1.1K\\
\textsf{zero-reversal-logs-final-logs6} & 134.1M & 22.0 & 14.1M & 14.3M & 33.1M & 595.9K & 87.9K & 185.8K & 22.3K & 14.6K\\
\textsf{OmpSCR-v2.0-c-QuickSort-134M-e} & 134.3M & 66.0 & 19.3M & 38.6M & 419.9K & 19.8K & 3.50 & 34.9K & 6.1K & 3.50\\
\textsf{OmpSCR-v2.0-c-QuickSort-134M-e1} & 134.3M & 226.0 & 19.8M & 92.3M & 419.9K & 20.7K & 33.92 & 35.0K & 6.6K & 33.92\\
\textsf{lufact} & 135.0M & 18.0 & 13.6M & 11.9M & 33.1M & 515.8K & 55.7K & 185.8K & 21.4K & 11.2K\\
\textsf{DRACC-DRACC-OMP-017-Counter-w1} & 135.0M & 68.0 & 58.4M & 39.7M & 15.00 & 0.00 & 5.50 & 15.00 & 0.00 & 5.50\\
\textsf{DRACC-DRACC-OMP-015-Counter-wr} & 135.0M & 68.0 & 35.9M & 39.6M & 5.1M & 154.2K & 1.5M & 16.00 & 1.27 & 6.25\\
\textsf{DataRaceBench-DRB148-critical1} & 135.0M & 68.0 & 36.1M & 39.3M & 5.4M & 161.3K & 1.6M & 16.00 & 2.36 & 5.75\\
\textsf{DRACC-DRACC-OMP-010-Counter-wr} & 135.0M & 68.0 & 35.9M & 39.5M & 5.2M & 155.1K & 1.5M & 16.00 & 1.82 & 5.58\\
\textsf{DRACC-DRACC-OMP-009-Counter-wr} & 135.0M & 68.0 & 36.1M & 38.7M & 5.3M & 160.9K & 1.5M & 16.00 & 2.09 & 4.50\\
\textsf{DRACC-DRACC-OMP-016-Counter-wr} & 135.0M & 68.0 & 35.9M & 39.6M & 5.4M & 161.7K & 1.6M & 16.00 & 1.55 & 5.43\\
\textsf{OmpSCR-v2.0-c-LUreduction-136M} & 136.4M & 66.0 & 10.7M & 38.7M & 42.2M & 1.2M & 175.1K & 89.4K & 71.9K & 1.8K\\
\textsf{OmpSCR-v2.0-c-LUreduction-137M} & 136.9M & 226.0 & 86.3M & 93.1M & 37.2M & 1.2M & 2.7M & 89.1K & 64.3K & 22.9K\\
\textsf{DataRaceBench-DRB144-critical} & 140.0M & 68.0 & 43.0M & 39.9M & 16.00 & 0.87 & 4.52 & 16.00 & 0.87 & 4.52\\
\textsf{OmpSCR-v2.0-cpp-sortOpenMP-cpp} & 141.7M & 68.0 & 96.6M & 39.6M & 16.00 & 0.55 & 2.50 & 16.00 & 0.55 & 2.50\\
\textsf{HPCBench-OmpSCR-v2.0-c-fft6-14} & 146.0M & 66.0 & 56.0M & 39.0M & 30.00 & 0.69 & 5.25 & 30.00 & 0.69 & 5.25\\
\textsf{OmpSCR-v2.0-c-fft6-146M-events} & 146.0M & 226.0 & 56.1M & 94.5M & 74.00 & 7.75 & 40.17 & 74.00 & 7.75 & 40.17\\
\textsf{OmpSCR-v2.0-c-Pi-150M-events-1} & 150.0M & 66.0 & 103.0M & 39.3M & 27.00 & 7.89 & 4.85 & 27.00 & 7.89 & 4.85\\
\textsf{OmpSCR-v2.0-c-Pi-150M-events-5} & 150.0M & 226.0 & 149.6M & 95.9M & 91.00 & 84.45 & 25.13 & 91.00 & 84.45 & 25.13\\
\textsf{HPCBench-NPBS-IS.W-153M-events} & 152.9M & 66.0 & 37.9M & 39.4M & 51.4M & 2.6M & 12.4K & 2.0M & 553.2K & 11.55\\
\textsf{OmpSCR-v2.0-cpp-sortOpenMP-cpp1} & 164.0M & 68.0 & 65.6M & 40.8M & 667.6K & 19.8K & 7.2K & 99.0K & 6.9K & 1.3K\\
\textsf{OMPRacer-amg2013-170M-events-1} & 169.9M & 358.0 & 21.8M & 130.3M & 24.0M & 703.6K & 4.3M & 595.2K & 43.7K & 68.8K\\
\textsf{SimpleMOC-170M-events-16-threa} & 170.2M & 68.0 & 43.1M & 41.0M & 25.6K & 723.89 & 28.13 & 2.1K & 55.38 & 26.27\\
\textsf{HPCBench-graph500-171M-events} & 171.3M & 66.0 & 49.2M & 40.0M & 86.0M & 2.6M & 340.3K & 115.2K & 18.3K & 11.5K\\
\textsf{HPCBench-graph500-172M-events} & 172.5M & 226.0 & 49.8M & 102.8M & 83.9M & 2.8M & 4.0M & 118.5K & 18.0K & 19.4K\\
\textsf{CoMD-CoMD-omp-task-deps-174M-e} & 174.1M & 66.0 & 24.1M & 39.1M & 124.3M & 3.4M & 94.0K & 16.5K & 15.8K & 1.9K\\
\textsf{CoMD-CoMD-openmp-174M-events-1} & 174.1M & 66.0 & 24.7M & 40.1M & 124.3M & 3.7M & 95.9K & 16.5K & 15.9K & 1.8K\\
\textsf{CoMD-CoMD-openmp-175M-events-5} & 175.1M & 226.0 & 30.4M & 103.4M & 128.5M & 3.5M & 3.5M & 16.9K & 15.5K & 6.2K\\
\textsf{CoMD-CoMD-omp-task-175M-events} & 175.1M & 226.0 & 29.9M & 103.4M & 128.5M & 3.1M & 3.4M & 16.9K & 15.7K & 6.2K\\
\textsf{CoMD-CoMD-omp-task-174M-events} & 175.1M & 226.0 & 30.5M & 103.3M & 128.5M & 3.1M & 3.3M & 16.9K & 15.6K & 6.3K\\
\textsf{CoMD-CoMD-omp-task-deps-175M-e} & 175.1M & 226.0 & 30.5M & 103.3M & 128.5M & 3.5M & 3.3M & 16.9K & 15.6K & 6.3K\\
\textsf{DataRaceBench-DRB062-matrixvec} & 183.9M & 66.0 & 90.3M & 38.7M & 33.9M & 1.1M & 3.5M & 31.00 & 16.82 & 19.53\\
\textsf{OMPRacer-amg2013-190M-events-5} & 189.6M & 518.0 & 25.1M & 159.6M & 28.0M & 928.1K & 8.9M & 718.4K & 65.3K & 155.0K\\
\textsf{OmpSCR-v2.0-c-LoopsWithDepende} & 192.6M & 66.0 & 25.0M & 40.5M & 2.8K & 88.75 & 52.42 & 32.00 & 15.82 & 17.45\\
\textsf{DataRaceBench-DRB062-matrixvec1} & 193.2M & 226.0 & 99.5M & 106.9M & 36.0M & 1.4M & 17.5M & 116.00 & 111.07 & 72.32\\
\textsf{OmpSCR-v2.0-c-MolecularDynamic} & 204.3M & 66.0 & 77.8M & 40.8M & 83.2M & 2.5M & 51.7K & 1.6K & 1.6K & 814.13\\
\textsf{OmpSCR-v2.0-c-MolecularDynamic1} & 204.4M & 226.0 & 86.9M & 109.1M & 87.0M & 2.7M & 569.4K & 1.7K & 1.7K & 1.7K\\
\textsf{OMPRacer-miniFE-207M-events-58} & 206.7M & 518.0 & 49.8M & 165.7M & 27.8M & 1.3M & 6.3M & 992.2K & 51.0K & 22.9K\\
\textsf{OMPRacer-miniFE-208M-events-18} & 207.7M & 358.0 & 48.3M & 144.2M & 19.4M & 816.5K & 1.5M & 968.5K & 41.2K & 18.0K\\
\textsf{misc-graphchi-graphchi} & 215.8M & 86.0 & 51.7M & 52.0M & 1.8M & 42.2K & 1.2K & 318.5K & 10.1K & 83.08\\
\textsf{zero-reversal-logs-final-logs} & 217.5M & 38.0 & 95.8M & 24.7M & 750.6K & 9.4K & 33.07 & 177.00 & 5.18 & 4.67\\
\textsf{misc-biojava-biojava} & 221.0M & 22.0 & 59.8M & 14.7M & 2.00 & 0.00 & 0.00 & 2.00 & 0.00 & 0.00\\
\textsf{HPCBench-HPCCG-228M-events-16} & 228.1M & 66.0 & 24.0M & 41.3M & 30.8M & 1.1M & 310.2K & 15.0K & 14.8K & 1.0K\\
\textsf{HPCBench-HPCCG-230M-events-56} & 229.5M & 226.0 & 35.5M & 113.3M & 41.8M & 2.0M & 3.6M & 15.8K & 15.7K & 13.8K\\
\textsf{DRB177-fib-taskdep-yes-382M-ev} & 236.2M & 70.0 & 115.8M & 43.7M & 2.1K & 60.58 & 274.45 & 1.2K & 56.64 & 226.92\\
\textsf{CoMD-CoMD-omp-taskloop-251M-ev} & 251.5M & 66.0 & 52.7M & 41.5M & 981.7K & 35.6K & 224.43 & 30.8K & 1.1K & 4.80\\
\textsf{CoMD-CoMD-omp-taskloop-251M-ev1} & 251.5M & 226.0 & 53.5M & 115.9M & 798.4K & 55.1K & 2.5K & 64.9K & 2.1K & 38.37\\
\textsf{misc-cassandra-cassandra} & 259.1M & 704.0 & 150.3M & 218.8M & 42.1K & 3.7K & 30.7K & 9.4K & 255.93 & 6.9K\\
\textsf{OmpSCR-v2.0-cpp-sortOpenMP-cpp2} & 295.5M & 226.0 & 167.5M & 120.9M & 144.5M & 4.5M & 1.3M & 6.1K & 5.9K & 6.0K\\
\textsf{HPCBench-NPBS-IS.W-300M-events} & 300.1M & 226.0 & 60.5M & 121.6M & 117.2M & 12.1M & 81.3K & 2.1M & 597.6K & 324.48\\
\textsf{zero-reversal-logs-final-logs1} & 307.3M & 42.0 & 282.9M & 27.7M & 10.4M & 401.6K & 614.85 & 1.00 & 0.62 & 0.88\\
\textsf{tsp} & 312.0M & 38.0 & 270.1M & 25.2M & 10.4M & 430.7K & 480.34 & 1.00 & 0.52 & 0.90\\
\textsf{OmpSCR-v2.0-c-LoopsWithDepende1} & 337.2M & 226.0 & 180.8M & 123.5M & 7.3K & 239.65 & 2.0K & 118.00 & 54.33 & 75.37\\
\textsf{OmpSCR-v2.0-c-LoopsWithDepende2} & 337.3M & 226.0 & 181.0M & 124.8M & 8.5K & 274.67 & 2.0K & 116.00 & 55.82 & 77.00\\
\textsf{pmd} & 367.0M & 54.0 & 104.5M & 35.4M & 144.8K & 3.5K & 0.96 & 1.6K & 50.43 & 0.96\\
\textsf{DRB177-fib-taskdep-yes-211M-ev} & 382.1M & 70.0 & 187.4M & 45.3M & 5.7K & 188.85 & 488.33 & 1.8K & 164.80 & 353.00\\
\textsf{OmpSCR-v2.0-c-LoopsWithDepende3} & 394.0M & 226.0 & 207.5M & 128.8M & 9.6K & 250.02 & 483.77 & 118.00 & 55.55 & 75.58\\
\textsf{OmpSCR-v2.0-c-LoopsWithDepende4} & 394.0M & 226.0 & 205.3M & 128.7M & 9.4K & 127.56 & 370.32 & 113.00 & 46.31 & 65.63\\
\textsf{zero-reversal-logs-final-logs2} & 397.8M & 24.0 & 297.6M & 16.0M & 1.00 & 0.02 & 0.00 & 1.00 & 0.02 & 0.00\\
\textsf{montecarlo} & 494.0M & 18.0 & 102.0M & 12.3M & 57.2K & 1.6K & 129.52 & 5.00 & 1.15 & 1.00\\
\textsf{OmpSCR-v2.0-c-fft-496M-events} & 496.0M & 70.0 & 115.8M & 45.9M & 2.2M & 62.2K & 28.10 & 2.0M & 57.5K & 1.73\\
\textsf{OmpSCR-v2.0-c-fft-496M-events1} & 496.0M & 230.0 & 102.6M & 135.9M & 2.1M & 63.8K & 302.18 & 2.1M & 63.8K & 18.53\\
\textsf{OMPRacer-Lulesh-543M-events-16} & 543.4M & 70.0 & 91.2M & 46.2M & 119.8M & 3.6M & 209.5K & 326.6K & 203.3K & 840.27\\
\specialrule{1pt}{0pt}{0pt}
\end{tabular*}
\end{adjustbox}
}
\label{tab:time3}
\end{table*}

%% file: tables/table4.tex

\begin{table*}[t!]
\label{table:table4}

\centering
\scalebox{0.63}{
\begin{adjustbox}{center}
\renewcommand{\arraystretch}{1.0}
\begin{tabular*}{1.57\columnwidth}{!{\VRule[1pt]}c!{\VRule[1pt]}c|c!{\VRule[1pt]}c|c!{\VRule[1pt]}c|c|c!{\VRule[1pt]}c|c|c!{\VRule[1pt]}}
\specialrule{1pt}{0pt}{0pt}
1 & 2 & 3 & 4 & 5 & 6 & 7 & 8 & 9& 10 & 11\\
\specialrule{1pt}{0pt}{0pt}
\rowcolor[HTML]{EFEFEF} 
\cellcolor[HTML]{EFEFEF} Benchmark
& \cellcolor[HTML]{EFEFEF} {trace length}
& \cellcolor[HTML]{EFEFEF} $M$
& \cellcolor[HTML]{EFEFEF} {\pacer sampled}
& \cellcolor[HTML]{EFEFEF}{\RND sampled}
& \multicolumn{3}{c!{\VRule[1pt]}}{\cellcolor[HTML]{EFEFEF}{\# Warnings}}

& \multicolumn{3}{c!{\VRule[1pt]}}{\cellcolor[HTML]{EFEFEF}{\# Warning variables }} \\
\cmidrule[0.5pt]{6-11}

\rowcolor[HTML]{EFEFEF} 
 \cellcolor[HTML]{EFEFEF} 
& \cellcolor[HTML]{EFEFEF} 
& \cellcolor[HTML]{EFEFEF} 
& \cellcolor[HTML]{EFEFEF} 
& \cellcolor[HTML]{EFEFEF} 
& \cellcolor[HTML]{EFEFEF} \textsc{FT}
& \cellcolor[HTML]{EFEFEF} \pacer
& \cellcolor[HTML]{EFEFEF} \RND
& \cellcolor[HTML]{EFEFEF} \textsc{FT}
& \cellcolor[HTML]{EFEFEF} \pacer
& \cellcolor[HTML]{EFEFEF} \RND\\

\specialrule{1pt}{0pt}{0pt}

\textsf{misc-zxing-zxing} & 546.4M & 64.0 & 97.8M & 42.4M & 10.1M & 264.2K & 9.40 & 27.2K & 754.49 & 1.58\\
\textsf{OMPRacer-Lulesh-569M-events-56} & 569.5M & 230.0 & 114.1M & 138.3M & 145.3M & 4.1M & 2.5M & 470.3K & 309.7K & 21.8K\\
\textsf{luindex} & 570.0M & 24.0 & 298.8M & 16.3M & 1.00 & 0.07 & 0.00 & 1.00 & 0.07 & 0.00\\
\textsf{zero-reversal-logs-final-logs3} & 606.9M & 22.0 & 51.7M & 15.0M & 0.00 & 0.00 & 0.00 & 0.00 & 0.00 & 0.00\\
\textsf{sor} & 608.0M & 18.0 & 49.0M & 12.3M & 0.00 & 0.00 & 0.00 & 0.00 & 0.00 & 0.00\\
\textsf{OmpSCR-v2.0-c-LoopsWithDepende2} & 112.6M & 66.0 & 14.2M & 37.4M & 2.9K & 57.60 & 106.07 & 30.00 & 14.16 & 19.00\\
\textsf{OmpSCR-v2.0-c-LoopsWithDepende3} & 112.6M & 66.0 & 14.1M & 37.0M & 2.9K & 70.86 & 89.90 & 30.00 & 13.98 & 18.93\\
\textsf{OmpSCR-v2.0-cpp-sortOpenMP-cpp6} & 114.2M & 228.0 & 75.4M & 85.3M & 1.1M & 72.6K & 133.0K & 99.7K & 9.3K & 15.7K\\
\textsf{OmpSCR-v2.0-cpp-sortOpenMP-cpp7} & 115.4M & 228.0 & 83.9M & 84.9M & 56.00 & 2.00 & 40.58 & 56.00 & 2.00 & 40.58\\
\textsf{OmpSCR-v2.0-c-Mandelbrot-116M} & 115.7M & 66.0 & 55.8M & 37.4M & 26.00 & 0.95 & 5.03 & 26.00 & 0.95 & 5.03\\
\textsf{OmpSCR-v2.0-c-Mandelbrot-116M1} & 115.7M & 226.0 & 115.7M & 85.6M & 87.00 & 55.98 & 28.78 & 87.00 & 55.98 & 28.78\\
\textsf{OMPRacer-Kripke-117M-events-16} & 117.5M & 68.0 & 13.1M & 38.6M & 12.2M & 998.9K & 91.4K & 175.1K & 68.0K & 22.8K\\
\textsf{DRB177-fib-taskdep-yes-618M-ev} & 618.3M & 70.0 & 303.0M & 46.5M & 9.6K & 297.89 & 522.72 & 2.2K & 251.45 & 375.13\\
\textsf{DataRaceBench-DRB176-fib-taskd} & 618.3M & 70.0 & 303.9M & 46.5M & 7.9K & 203.64 & 489.02 & 2.4K & 174.24 & 353.95\\
\textsf{DRB177-fib-taskdep-yes-618M-ev1} & 618.3M & 230.0 & 303.0M & 139.8M & 27.5K & 1.1K & 4.6K & 6.8K & 934.60 & 2.5K\\
\textsf{DRB176-fib-taskdep-no-341M-eve} & 618.3M & 230.0 & 303.6M & 139.8M & 49.3K & 1.8K & 9.2K & 9.7K & 1.4K & 4.3K\\
\textsf{OmpSCR-v2.0-c-LoopsWithDepende5} & 674.2M & 226.0 & 349.6M & 137.1M & 6.7K & 198.91 & 187.25 & 116.00 & 50.64 & 64.22\\
\textsf{xalan} & 1000.0M & 60.0 & 225.0M & 40.5M & 2.3K & 7.77 & 8.98 & 202.00 & 7.27 & 0.82\\
\textsf{DRB177-fib-taskdep-yes-552M-ev} & 1.0B & 70.0 & 490.6M & 47.2M & 7.6K & 240.22 & 241.97 & 2.1K & 209.13 & 199.43\\
\textsf{OMPRacer-RSBench-1.2B-events-1} & 1.3B & 66.0 & 828.8M & 44.7M & 22.00 & 0.78 & 0.28 & 22.00 & 0.78 & 0.28\\
\textsf{OMPRacer-RSBench-1.2B-events-5} & 1.3B & 226.0 & 882.0M & 146.9M & 95.00 & 5.95 & 12.23 & 95.00 & 5.95 & 12.23\\
\textsf{DRB177-fib-taskdep-yes-1.6B-ev} & 1.6B & 70.0 & 793.8M & 47.6M & 11.4K & 374.89 & 258.20 & 2.4K & 306.84 & 209.53\\
\textsf{DRB176-fib-taskdep-no-1.6B-eve} & 1.6B & 70.0 & 795.7M & 47.6M & 12.5K & 396.09 & 319.93 & 2.9K & 309.67 & 245.15\\
\textsf{DRB176-fib-taskdep-no-1.6B-eve1} & 1.6B & 230.0 & 795.4M & 151.2M & 71.6K & 2.5K & 5.4K & 11.1K & 1.9K & 3.0K\\
\textsf{moldyn} & 1.7B & 18.0 & 650.1M & 12.4M & 17.6M & 668.5K & 747.57 & 18.4K & 17.5K & 3.24\\
\textsf{OmpSCR-v2.0-c-fft-2.1B-events} & 2.1B & 230.0 & 477.8M & 152.9M & 8.1M & 233.4K & 56.27 & 7.8M & 233.4K & 4.72\\
\textsf{avrora} & 2.4B & 32.0 & 956.3M & 22.0M & 3.1M & 93.0K & 20.0K & 414.7K & 17.5K & 4.4K\\
\textsf{raytracer} & 2.8B & 18.0 & 1.5B & 12.4M & 7.00 & 1.42 & 0.00 & 4.00 & 1.19 & 0.00\\

\specialrule{1pt}{0pt}{0pt}
\end{tabular*}
\end{adjustbox}
}
\label{tab:time4}
\end{table*}